\documentclass{article}
\usepackage[utf8x]{inputenc}
\usepackage[T1]{fontenc}
\usepackage{amsthm}
\usepackage{amsmath}
\usepackage{amssymb}
\usepackage{fullpage}
\usepackage{tikz}
\usetikzlibrary{positioning}

\renewcommand{\S}{{\mathcal{S}}}
\newcommand{\T}{{\mathcal{T}}}
\newcommand{\nS}{{\overline{S}}}
\newcommand{\nT}{{\overline{T}}}
\newcommand{\nSp}{\overline{S'}}

\newcommand{\ox}{\overline{x}}
\newcommand{\oy}{\overline{y}}

\newcommand{\fone}{\ensuremath{{\mathsf{F1}}}}
\newcommand{\ftwo}{\ensuremath{{\mathsf{eF2}}}}
\newcommand{\fthree}{\ensuremath{{\mathsf{eF3}}}}
\newcommand{\ffour}{\ensuremath{{\mathsf{eF4}}}}
\newcommand{\oftwo}{\ensuremath{{\mathsf{oF2}}}}

\newcommand{\mbm}{\ensuremath{{\mathsf{mbM}}}}

\newcommand{\blf}{\ensuremath{{\mathbf{s}_0}}}
\newcommand{\blfp}{\ensuremath{{\mathbf{s}_0^\prime}}}
\newcommand{\st}{{\mathbf{s}}}
\newcommand{\defi}{{\mathsf{def}}}
\newcommand{\floor}[1]{{\lfloor{#1}\rfloor}}
\newcommand{\ceil}[1]{{\lceil{#1}\rceil}}
\newcommand{\ignore}[1]{}

\newcommand{\mone}{\ensuremath{{\mathsf{M2}}}}
\newcommand{\mtwo}{\ensuremath{{\mathsf{M1}}}}
\newcommand{\mthree}{\ensuremath{{\mathsf{M3}}}}

\newtheorem{theorem}{Theorem}
\newtheorem{definition}{Definition}
\newtheorem{prop}[theorem]{Proposition}

\newtheorem{lemma}[theorem]{Lemma}

\begin{document}

\pagestyle{plain}

\raggedbottom

\title{Minority Becomes Majority in Social Networks%
\thanks{This work was partially supported by the COST Action IC1205 ``Computational Social Choice'',
by the Italian MIUR under the PRIN 2010-2011 project ARS TechnoMedia -- Algorithmics for Social Technological Networks,
by the European Social Fund and Greek national funds through the research funding program Thales on ``Algorithmic Game Theory'', by the EU FET project MULTIPLEX 317532, and by a Caratheodory basic research grant from the University of Patras.}
}

\author{Vincenzo Auletta\thanks{Universit\`a degli Studi di Salerno \texttt{auletta@unisa.it}}
\and Ioannis Caragiannis\thanks{CTI ``Diophantus'' \& University of Patras \texttt{caragian@ceid.upatras.gr}}
\and Diodato Ferraioli\thanks{Universit\`a degli Studi di Salerno \texttt{dferraioli@unisa.it}}
\and Clemente Galdi\thanks{Universit\`a di Napoli ``Federico II'' \texttt{clemente.galdi@unina.it}}
\and Giuseppe Persiano\thanks{Universit\`a degli Studi di Salerno \texttt{pino.persiano@unisa.it}}
}

\date{}

\maketitle

\setcounter{footnote}{0}

\begin{abstract}
It is often observed that agents tend to imitate the behavior of 
their neighbors in a social network. 
This imitating behavior might lead to the strategic decision of 
adopting a public behavior that differs from what the agent 
believes is the right one and this can subvert the behavior of the 
population as a whole.

In this paper, we consider the case in which agents express preferences 
over two alternatives and model social pressure with 
the {\em majority} dynamics: at each step an agent is selected and its 
preference is replaced by the majority of the preferences of her neighbors.
In case of a tie, the agent does not change her current preference.
A profile of the agents' preferences is {\em stable} if
the preference of each agent coincides with the preference of at least 
half of the neighbors (thus, the system is in equilibrium).

We ask whether there are network topologies that are robust to social 
pressure.
That is, we ask whether there are graphs in which the majority of preferences
in an initial profile $\st$ always coincides with the majority of 
the preference in all stable profiles reachable from $\st$. 
We completely characterize the graphs with this 
robustness property by showing that this is possible only if the 
graph has no edge or is a clique or very close to a clique.
In other words, except for this handful of graphs, every graph admits
at least one initial profile of preferences in which the majority dynamics
can subvert the initial majority.
We also show that deciding whether a graph admits a minority that becomes 
majority is NP-hard when the minority size is
at most $1/4$-th of the social network size.
%
\end{abstract}

\section{Introduction}
Social scientists are greatly interested in understanding how 
social pressure can influence the behavior of agents in a social network. 
We consider the case in which agents connected through a social network 
must choose between two alternatives and, for concreteness, we consider
two competing technologies: 
the current (or old) technology and a new technology.
To make their decision, the agents take into account two factors: 
their personal relative valuation of the two technologies and the 
opinions expressed by their social neighbors. 
Thus, the public action taken by an agent 
(i.e., adopting the new technology or staying with the old) 
is the result of a mediation between her personal valuation and the 
social pressure derived from her neighbors.
 
The first studies concerning the adoption of new technologies
date back to the middle of 20-th century, with the analysis of 
the adoption of hybrid seed corn among farmers in Iowa~\cite{ryan1950acceptance}
and of tetracycline by physicians in US~\cite{coleman1966medical}.

We assume that agents receive an initial signal about the quality of 
the new technology that constitutes the agent's initial preference. 
This signal is independent from the agent's social network;
e.g., farmers acquired information about the hybrid corn from salesman and physicians acquired information about tetracycline from scientific publications. 
After the initial preference is formed, 
an agent tends to conform her preference to the one of her neighbors and thus to \emph{imitate} 
their behavior, even if this disagrees with her own initial preference.
This imitating behavior can be explained in several ways:
an agent that sees a majority agreeing on an opinion
might think that her neighbors have access to some information unknown to her and hence 
they have made the better choice;
also agents can directly benefit from adopting the same behavior as their friends 
(e.g., prices going down).

Thus, the natural way of modeling the evolution of preferences in 
networks is through a majority dynamics: each agent has an initial 
preference and at each time step a subset of agents updates their 
opinion conforming to the majority of their neighbors in the network. 
As a tie-breaking rule it is usual to assume that when 
exactly half of the neighbors adopted the new technology, the agent 
decides to stay with her current choice to avoid the cost of a change. 
Thus, the network undergoes an opinion formation process where 
agents continue to update their opinions until a stable profile is 
reached, where each agent's behavior agrees with the majority of her neighbors.
Notice that the dynamics does not take into account the relative merits of the two
technologies and, without loss of generality, 
we adopt the convention that the technology 
that is preferred by the majority of the agents in the initial 
preference profile is the new technology.

In the setting described above, it is natural to ask whether and when 
the social pressure of conformism can change the opinion of some of the 
agents so that the initial majority is subverted.
In the case of the adoption of a new technology, 
we are asking whether a minority of agents supporting the old technology 
can orchestrate a campaign and convince enough agents to reject 
the new technology, even if the majority of the agents had initially preferred the new 
technology.

This problem has been extensively studied in the literature.
If we assume that updates occur \emph{sequentially}, 
one agent at each time step, then it is easy to design graphs (e.g., a star) where the 
old technology, supported by an arbitrarily small minority of agents, 
can be adopted by most of the agents. Berger \cite{berger} proved that 
such a result holds even if at each time step all agents \emph{concurrently} update their actions.
However, Mossel et al. \cite{mossel} and Tamuz and Tessler \cite{tamuz} 
proved that there are graphs for which, both with concurrent and sequential updates, at the end of the update 
process the new technology will be adopted by the majority of agents
with high probability.

In \cite{mossel,feldman} it is also proved that when the 
graph is an expander, agents will reach a 
\emph{consensus} on the new technology with high probability for both sequential and concurrent updates
(the probability is taken on the choice of initial configuration with a majority of 
new technology adopters). 
Thus, expander graphs are particularly efficient in aggregating 
opinions since, with high probability, social pressure does not prevent the diffusion of the new technology.

In this paper, 
we will extend this line of research by taking a worst-case approach instead of a probabilistic one.
We ask whether there are graphs that are {robust} to social pressure,
even when it is driven by a carefully and adversarially designed campaign.  
Specifically, 
we want to find out whether there are graphs in which no subset of the agents preferring the old technology
(and thus consisting of less than half of the agents) 
can manipulate the rest of the agents and drive the network to a stable profile in which the 
majority of the agents prefers the old technology.
This is easily seen to hold for two extreme graphs: the clique and the graph with no edge.
In this paper, we prove that these are essentially\footnote{It turns out that for an even number of nodes, 
there are a few more very dense graphs enjoying such a property.}
the only graphs where social pressure cannot subvert the majority.

In particular, our results highlight that even for expander graphs,
where it is known that agents converge with high probability to 
consensus on the new technology, it is possible to fix a minority and 
orchestrate a campaign that brings the network into a stable profile where at least half of the agents decide to not adopt the new technology.

\smallskip \noindent \textit{Overview of our contribution.}
%
We consider the following sequential dynamics. 
We have $n$ agents and at any given point the system is described by the 
profile $\st$ in which $\st(i)\in\{0,1\}$ is the preference of the $i$-th agent.
We say that agent $i$ is {\em unhappy} in profile $\st$ if the majority of her 
neighbors have a preference different from $\st(i)$.
Profiles evolve according to the dynamics in which an {\em update} 
consists of non-deterministically selecting an unhappy agent and 
changing its preference.
A profile in which no agent is unhappy is called {\em stable}. 

In Section \ref{sec:char} (see Theorem~\ref{thm:forbiddennotmbm} and~\ref{thm:forbidden}),
we characterize the set of social networks (graphs) where a majority can 
be subverted by social pressure. More specifically, 
we show that for each of these graphs it is possible to select a minority 
of agents not supporting the new technology and 
a sequence of updates (a campaign) that leads the network to a stable 
profile where the majority of the agents prefers the old technology.
As described above, we will prove that this class is very large and contains all graphs except a small set of forbidden graphs,
consisting of the graph with no edges and of other 
graphs that are almost cliques.
Proving this fact turned out to be a technically challenging task and 
it is heavily based on properties of local optima of graph bisections.

Then we turn our attention to related computational questions.
First we show that
we can compute in polynomial time an initial preference profile, 
where the majority of the agents supports the new technology, 
and a sequence of update that ends in a stable profile where at 
least half of the agents do not adopt the new 
technology. This is done through a polynomial-time local-search computation of a bisection of locally minimal width.

We actually prove a stronger result. 
In principle, it could be that from the starting profile the 
system needs to undergo a long sequence of updates, in which the minority 
gains and loses member to eventually reach a stable profile in which the
minority has become a majority. Our algorithm shows that this can always
be achieved by means of a short sequence of at most two updates
after which any sequence of updates will bring the system to a stable
profile in which the initial minority has become majority.
This makes the design of an adversarial campaign even more realistic,
since such a campaign only has to identify the few ``swing'' agents and thus
it turns out to be very simple to implement.

However, the simplicity of the subverting campaign comes at a cost.
Indeed, our algorithm always computes an initial preferences profile that 
has very large minorities, consisting of $\left\lfloor\frac{n-1}{2}\right\rfloor$ agents.
We remark that, even in case of large minorities, it is not trivial
to give a sequence of update steps that ends in a stable profile where the majority is subverted.
Indeed, even if the large minority of the original profile makes it easy to find a few agents of the original majority that prefer to change their opinions, this is not sufficient in order to prove that the majority has been subverted,
since we have also to prove that there are no other nodes in the original minority that prefer to change their preference. 

Moreover, we observe that, 
even if there are cases in which such a large minority is necessary,
the idea behind our algorithm can be easily turned into an heuristic
that checks whether the majority can be subverted
by
a smaller minority
(e.g., by considering unbalanced partitions in place of bisections).

On the other side, we show that a large size of the minority in the 
initial preference profile seems to be necessary in order to quickly compute 
a subverting minority and its corresponding sequence of updates.
Indeed, given a $n$-node social network, deciding whether there exists 
a minority of less than $n/4$ nodes and a sequence of update steps that bring the system to a stable profile in which the majority has been subverted is an NP-hard problem (see Theorem~\ref{thm:reduction}).

The main source of computational hardness seems to arise from the computation of initial preference profile.
Indeed, if this profile is given, computing the maximum number of adopters of the new technology (and, hence, deciding whether majority can be subverted) and the corresponding sequence of updates turns out to be 
possible in polynomial time (see Theorem~\ref{thm:poly-mbm}).

\smallskip \noindent \textit{Related work.}
There is a vast literature on the effect that social pressure has on the behavior of a system as a whole.
In many works, 
influence is modeled by agents simply following the majority \cite{berger,mossel,tamuz,feldman}.
A generalization of this imitating behavior is discussed in \cite{mossel}.

A different approach is taken in \cite{mossel-sly}, where each agent updates her behavior
according to a Bayes rule that takes in account its own initial preference and what is declared by neighbors
on the network.

Yet another approach assumes that agents are strategic and rational.
That is, they try to maximize some utility function that depends on
the level of coordination with the neighbors on the network.
Here, the updates occur according to a best response dynamics or
some other more complex game dynamics.
Along this direction, particularly relevant to our works
are the ones considering best-response dynamics from truthful profiles
in the context of iterative voting, e.g., see \cite{MPRJ10} and \cite{BCMP13}.
In particular, closer to our current work is the paper of Br\^anzei et al. \cite{BCMP13}
who present bounds on the quality of equilibria that can be reached from a truthful
profile using best-response play and different voting rules.
The important difference is that there is no underlying network in their work.

Our work is also strictly related with a line of work in social sciences
that aims to understand how opinions are formed and expressed in a social context. 
A classical simple model in this context has been proposed by Friedkin and Johnsen \cite{FJ90} (see also~\cite{D74}). 
Its main assumption is that each individual has a private initial belief and that the
opinion she eventually expresses is the result of a repeated averaging between her 
initial belief and the opinions expressed by other individuals with whom she has social relations.
The recent work of Bindel et al. \cite{BKO11} assumes that initial beliefs and opinions belong to $[0,1]$ and 
interprets the repeated averaging process as a best-response play in a naturally defined 
game that leads to a unique equilibrium.

An obvious refinement of this model is to consider discrete initial beliefs and opinions
by restricting them, for example, to two discrete values
(see \cite{fgvSAGT12} and \cite{ckoEC13}).
Clearly, the discrete nature of the opinions does not allow for averaging anymore and 
several nice properties of the opinion formation models mentioned above --- such as the uniqueness of the outcome --- are lost. 
In contrast, in \cite{fgvSAGT12} and \cite{ckoEC13}, it is assumed that each agent is strategic and aims to pick the most beneficial strategy for her,
given her internal initial belief and the strategies of her neighbors.
Interestingly, it turns out that the majority rule used in this work
for describing how agents update their behavior
can be seen as a special case of the discrete model of \cite{fgvSAGT12} and \cite{ckoEC13},
in which agents assign a weight to the initial preference smaller than the one given to
the opinion of the neighbors.

Studies on social networks consider several phenomena related to the spread of social influence
such as information cascading, network effects, epidemics, and more.
The book of Easley and Kleinberg \cite{EK10} provides an excellent introduction
to the theoretical treatment of such phenomena. From a different perspective,
problems of this type have also been considered in the distributed computing literature,
motivated by the need to control and restrict the influence of failures in distributed systems;
e.g., see the survey by Peleg \cite{Peleg02} and the references therein.

\smallskip \noindent \emph{Preliminaries.}
We formally describe our model as follows.
There are $n$ agents; we use $[n]=\{1, 2, ..., n\}$ to denote their set. 
Each agent corresponds to a distinct node of a graph $G=(V,E)$ that 
represents the {\em social network}; 
i.e., the network of social relations between the agents. 
Agent $i$ has an initial preference $\blf(i)\in\{0,1\}$.
At each time step, agent $i$ can update her preference to $\st(i) \in \{0,1\}$. 
A {\em profile} is a vector of preferences, with one preference per agent.
We use bold symbols for profiles; i.e., $\st = (\st(1),\ldots, \st(n))$.
In particular, we sometimes call the profile of initial preferences $(\blf(1),\ldots,\blf(n))$ 
as the \emph{truthful profile}.
Moreover, for any $y \in \{0,1\}$, we denote as $\oy$ the negation of $y$;
i.e., $\oy=1-y$.

%

A graph $G$ is \mbm\ ({\em minority becomes majority}) if
there exists a profile $\blf$ of initial preferences such that:
the number of nodes that prefer $0$ is a strict majority,
i.e., $|\{ x\in V \colon\blf(x) = 0\}| > n/2$;
and there is a {\em subverting} sequence of updates
that starts from $\blf$
and reaches a stable profile $\st$ in which the number of nodes that prefer $0$
is not a majority,
i.e., $|\{x\in V \colon \st(x)=0\}|\leq n/2$.
A profile of initial preferences that witnesses a graph being \mbm\ will be also termed \mbm.

\section{Characterizing the \mbm\ Graphs}\label{sec:char}
The main result  of this section is a characterization of the \mbm\ graphs.
More formally, we have the following definition.
\begin{definition}\label{def:evenforbidden}
A graph $G$ with $n$ nodes is {\em forbidden} 
if one of the following conditions is satisfied.\\
\textbf{\ \fone:} $G$ has no edge;\\
\textbf{\oftwo:}
$G$ has an odd number of nodes,
  all of degree $n-1$ (that is, $G$ is a clique);\\
\textbf{\ftwo:}
$G$ has an even number of nodes
    and all its nodes have degree at least $n-2$;\\
\textbf{\fthree:}
$G$ has an even number of nodes, 
    $n-1$ nodes of $G$ form a clique, and 
            the remaining node has degree at most $2$;\\
\textbf{\ffour:}
$G$ has an even number of nodes, 
    $n-1$ nodes of $G$ have degree $n-2$ but they do not form a clique, 
            and the remaining node has degree at most $4$.
\end{definition}

%
We begin by proving the following statement.
\begin{theorem}\label{thm:forbiddennotmbm}
No forbidden graph is \mbm.
\end{theorem}
\begin{proof}
We will distinguish between cases for a forbidden graph $G$. 
Clearly, if $G$ is \fone, then it is not {\mbm} since no 
node can change its preference. Now assume that $G$ is \ftwo\
(respectively, \oftwo) and consider a profile in which there are at least $\frac{n}{2}+1$ 
(respectively, $\frac{n+1}{2}$) agents with preference $0$. Then, every node $x$ with
initial preference $0$ has at most $\frac{n}{2}-1$ neighbors with initial preference $1$ 
and at least $\frac{n}{2}-1$ neighbors with initial preference $0$ 
(respectively, at most $\frac{n-1}{2}$ neighbors with initial preference $1$ 
and at least $\frac{n-1}{2}$ neighbors with initial preference $0$).
Hence, $x$ is not unhappy and stays with preference $0$.

Now, consider the case where $G$ is \fthree\ and let $u$ 
be the node of degree at most $2$. Consider profile $\blf$ of initial preferences 
in which there are at least $\frac{n}{2}+1$ agents with preference $0$. 
First observe that in the truthful profile $\blf$
any node $x$ other than $u$ that has preference $0$ 
is adjacent to at most 
$\frac{n}{2}-1$ nodes with initial preference $1$ and to at least 
$\frac{n}{2}-1$ nodes with initial preference $0$.
Then, $x$ is not unhappy and stays with preference $0$.
Hence, $u$ is the only node that may want to switch from $0$ to $1$.
But this is possible only if all nodes in the 
neighborhood of $u$ have preference $1$, which implies 
that the neighborhood of any node with initial preference $0$ does 
not change after the switch of $u$, i.e., nodes with 
preference $0$ still are not unhappy and thus they have no incentive to switch to $1$. 
Then, any node with preference $1$ that is not adjacent
to $u$ has at most $\frac{n}{2}-2$ neighbors 
with preference $1$ and at least $\frac{n}{2}$ neighbors 
with preference $0$. Also, any node with preference $1$ that 
is adjacent to $u$ has $\frac{n}{2}-1$ neighbors with 
preference $1$ and $\frac{n}{2}$ neighbors with preference 
$0$. So, every node with preference 
$1$ will eventually switches to $0$.

It remains to consider the case where $G$ is \ffour; 
let $u$ be the node of degree at most $4$. Actually, 
it can be verified that $u$ can have degree either $2$ 
or $4$ and its neighbors form pair(s) of non-adjacent 
nodes. Consider a truthful profile in which there are at least $\frac{n}{2}+1$ 
agents with preference $0$. Observe that a 
node different from $u$ that has initial preference $0$ has at 
most $\frac{n}{2}-1$ neighbors with preference $1$ and at least $\frac{n}{2}-1$ neighbors with preference 
$0$. So, it is not unhappy and has no incentive to switch to preference 
$1$. The only node that might do so is $u$, provided 
that the strict majority of its neighbors (i.e., both 
of them if $u$ has degree $2$ and at least three of 
them if $u$ has degree $4$) have preferences $1$. This 
switch cannot trigger another switch of the preference of an agent from $0$
node to $1$. Indeed, there is at most one agent with preference $0$
that can be adjacent to $u$. Since this node 
is not adjacent to one
of the neighbors of $u$ with preference $1$, it has at most 
$\frac{n}{2}-1$ neighbors with preference $1$ (and at least 
$\frac{n}{2}-1$ neighbors with preference $0$). Hence, it has 
no incentive to switch to preference $1$ either.
Now, consider two neighbors of $u$ with preference $1$ that 
are not adjacent (these nodes certainly exist). Each
of them is adjacent to $\frac{n}{2}-2$ nodes with 
preference $1$ and $\frac{n}{2}$ nodes with preference 
$0$. Hence, they have an incentive to switch to $0$.
Then, the number of nodes with preference $1$ is at most
$\frac{n}{2}-2$ and eventually all nodes will switch
to preference $0$.
\end{proof}

The following is the main result of this section. 
\begin{theorem}\label{thm:forbidden}
Every non-forbidden graph is \mbm.
\end{theorem}
We next give the proof for the simpler case of graphs with an odd number of vertices and postpone the full proof to Section~\ref{sec:even}.
Let us start with the following definitions.
A {\em bisection} $\S=(S,\nS)$ of a graph $G=(V,E)$ with $n$ nodes
is simply a partition of the nodes of $V$ into two sets $S$ and $\nS$ of sizes
$\lceil n/2\rceil$ and $\lfloor n/2\rfloor$, respectively.
We will refer to $S$ and $\nS$ as the {\em sides} of bisection $\S$.
The {\em width} $W(S,\nS)$ of a bisection $\S$ is the number of 
edges of $G$ whose endpoints belong to different sides of the partition.
The \emph{minimum} bisection $\S$ of $G$ has minimum width among
all partitions of $G$.
We extend notation $W(A,B)$ to any pair $(A,B)$ of subsets of nodes of 
$G$ in the obvious way. When $A=\{x\}$ is a singleton we will
write $W(x,B)$ and similarly for $B$.
Thus, if nodes $x$ and $y$ are adjacent, then $W(x,y)=1$; 
otherwise $W(x,y)=0$.
For a bisection $\S=(S,\nS)$, we define
the {\em deficiency $\defi_\S(x)$ of node $x$ w.r.t. 
bisection $\S$} as
$\defi_\S(x)=W(x,S)-W(x,\nS)$ if $x\in S$,
and $\defi_\S(x)=W(x,\nS)-W(x,S)$ if $x\in \nS$.

\begin{lemma}\label{lem:minimal-ICALP}
Let  $\S=(S,\nS)$ be a minimum bisection of a graph $G$ with $n$ nodes.
Then, for every $x\in S$ and $y\in\nS$, 
$\defi_\S(x)+\defi_\S(y)+2W(x,y)\geq 0.$
Moreover if $n$ is odd,
$\defi_\S(x)\geq 0.$
\end{lemma}
\begin{proof}
Set $A=S\setminus\{x\}$, $B=\nS\setminus\{y\}$, 
$T=A\cup\{y\}$ and $\overline{T}=B\cup\{x\}$.
Note that $W(T,\overline{T}) = W(A,B)+W(x,A)+W(y,B)+W(x,y)$ and
$W(S,\nS) = W(A,B)+W(x,B)+W(y,A)+W(x,y)$.
Then, by minimality,
$
0\leq W(T,\overline{T}) - W(S,\nS) = W(x,A)+W(y,B) - W(x,B)-W(y,A)
  =   W(x,S)-W(x,\nS)+W(y,\nS)-W(y,S)+2W(x,y)
  =   \defi_\S(x)+\defi_\S(y)+2W(x,y).
$
For the second part of the lemma, we consider partition 
$(\nS\cup\{x\},S\setminus\{x\})$.
\end{proof}

We have the following technical lemma.
\begin{lemma}\label{lem:tone-ICALP}
Suppose that a
graph $G$ admits a bisection $\S=(S,\nS)$ in 
which $S$ consists of nodes with non-negative deficiency and
includes at least one node with positive deficiency. Then $G$ is \mbm.
\end{lemma}
\begin{proof}
Let $v$ be the node with positive deficiency in $S$ and 
consider profile $\blf$ of initial preferences in which any node in $S$ except $v$
has preference $1$ and remaining nodes have preference $0$.
Hence, in $\blf$ there is a majority of $\lceil n/2\rceil$ agents with preference $0$.
Observe also that in $\blf$,
$v$ is adjacent to $W(v,S)$ nodes with preference $1$ and
to $W(v,\overline{S})$ nodes with preference $0$. 
Since $\defi_{\mathcal{S}}(v)>0$ then $v$ is unhappy with preference $0$ and updates her
preference to $1$.
We thus reach a profile $\st_1$ in which $\lceil n/2\rceil$ nodes
have preference $1$ (that is, all nodes in $S$).
We conclude the proof of the lemma by showing that 
every node of $S$ is not unhappy and thus it stays with preference $1$\footnote{This
is sufficient since the switch of nodes in $\overline{S}$ that are unhappy with preference $0$
only increases the number of nodes with preference $1$.
Moreover, if some nodes in $\overline{S}$ switch their preferences, then
the number of nodes with preference $1$ in the neighborhood of any node in $S$ can only increase.}.
This is obvious for $v$. Let us consider $u\in S$ and $u\ne v$.
Then $u$ has $W(u,S)$ neighbors with preference $1$ and 
$W(u,\overline{S})$ neighbors with preference $0$. 
Since $\defi_{\mathcal{S}}(u)\geq 0$, we have that $W(u,S)\geq W(u,\overline{S})$.
Hence, the number of neighbors of $u$ with preference $0$ is not a majority.
Then, $u$ is not unhappy, and thus stay with preference $1$.
\end{proof}

We are now ready to prove Theorem~\ref{thm:forbidden} for odd-sized graphs. 
We remind the reader that the (more complex) proof for even-size
graphs is in Section~\ref{sec:even}.
\begin{prop}
\label{prop:main}
Non-forbidden graphs with an odd number of nodes are \mbm.
\end{prop}
\begin{proof}
Let $G$ be a non-forbidden graph with an odd number of nodes
and 
let $\S=(S,\nS)$ be a minimum bisection for $G$. 
By Lemma~\ref{lem:minimal-ICALP}, we have
that $\defi_\S(x)\geq 0$, for all $x\in S$.
If $S$ contains at least a node $v$ with $\defi_\S(v)>0$ then, by Lemma~\ref{lem:tone-ICALP}, 
$G$ is $\mbm$. So assume that $\defi_\S(x)=0$ for all $x\in S$.

Lemma~\ref{lem:minimal-ICALP} implies that if $\defi_\S(v)<0$ for 
$v\in\nS$ then $\defi_\S(v)\geq -2$ and $v$ is connected to all 
vertices in $S$. Therefore $W(v,S)=\lceil n/2\rceil$ and, since 
$W(v,\nS)\leq\lfloor n/2\rfloor-1$, we conclude that $\defi_\S(v)=-2$.
We denote by $A$ the set of all the nodes $y\in\nS$ with
$\defi_\mathcal{S}(x)=-2$; therefore, 
all nodes $y\in \overline{S} \setminus A$ have $\defi_\mathcal{S}(y) \geq 0$.

Let us first consider the case in which $A \neq \emptyset$ and 
there are two non-adjacent nodes $u,w\in S$.
Then pick any node $v\in A$ and consider partition 
$\mathcal{T}=(T,\overline{T})$ with $T=S \cup\{v\}\setminus\{u\}$.
We have that $W(v,T)=W(v,S)-1=\lceil n/2\rceil -1$ and $W(v,\nT)=W(v,\nS)+1=\lfloor n/2\rfloor+1$ and hence $\defi_{\mathcal{T}}(v)=0$.
For any $x\in T\setminus\{v, w\}$, 
we have $\defi_{\mathcal{T}}(x)\geq \defi_\mathcal{S}(x)=0$.
Node $w$ is connected to $v$ but not to $u$ and, thus,
$\defi_{\mathcal{T}}(w)\geq \defi_{\mathcal{S}}(w)+2=2$.
Then, by Lemma~\ref{lem:tone-ICALP}, $G$ is \mbm.

Assume now that $A \neq \emptyset$ and $S$ is a clique.
That is, $W(x,S)=\lceil n/2\rceil -1$ for every $x\in S$, and, 
since $\defi_\S(x)=0$, it must be that $W(x,\nS)=W(x,S)$
and thus $x$ is connected to all nodes in $\nS$. 
Therefore, for all $y\in\nS$, $W(y,S)=\lceil n/2\rceil$ and, since 
$\defi_\S(y)\geq -2$ it must be 
that $W(y,\nS)\geq\lceil n/2\rceil-2=|\nS|-1$. 
In other words, every node of $\nS$ is connected to every node of $\nS$
and thus $G$ is a clique.

Finally, assume that $A=\emptyset$;
that is, $\defi_\mathcal{S}(y)\geq 0$ for any $y\in\nS$.
If for some $v \in \overline{S}$, we have $\defi_\S(v)>0$,
then consider partition $\T=(T,\nT)$ 
with $T=\nS\cup\{u\}$,
where $u$ is any node from $S$.
For any $x\in T\cap\nS$, $\defi_\T(x) \geq \defi_\S(x)\geq 0$,
$\defi_\T(u) = - \defi_S(u)=0$and 
$\defi_\T(v) \geq \defi_\S(v)\geq 1$. 
By Lemma~\ref{lem:tone-ICALP}, $G$ is \mbm.

Finally, we consider the case in which $\defi_\S(y)=0$ for every node $x$ of $G$.
Since $G$ is not empty, there exists at least one edge in $G$ and,
since the endpoints of this edge have $\defi_\S=0$ there must be at least
 node $v\in S$ with a neighbor $w\in\nS$.
Now, consider partition $\mathcal{T}=(T,\nT)$ with $T=S\cup\{w\}$.
We have that every node $x \in T\cap S$, has $\defi_\T(x)\geq\defi_\S(x)=0$,
$\defi_\T(w) = - \defi_\S(w) = 0$,
and $\defi_\T(w)>\defi_\S(w)=0$. The claim again follows by Lemma~\ref{lem:tone-ICALP}.
\end{proof}

We note that the only property required for invoking Lemma \ref{lem:minimal-ICALP} is local minimality.
Since a local-search algorithm can compute a locally minimal bisection in polynomial time,
we can make constructive the proof of Proposition~\ref{prop:main},
and quickly compute the subverting minority and the corresponding updates.

\section{Characterization for Even-Sized Graphs}
\label{sec:even}
We partition the set of all graphs with an even number $n$ of nodes 
into two sets:
the {\em extremal} graphs 
(these are graphs with an even number of nodes 
in which at least $n-1$ nodes have degree at least $n-2$), and
the {\em non-extremal} graphs 
(these are graphs with an even number of nodes 
in which at least two nodes have degree less than $n-2$).
Proposition~\ref{prop:extremal} (see Section~\ref{sec:extremal}) 
shows that all extremal graphs are \mbm.
Then,
Proposition~\ref{prop:nonextremal} (see Section~\ref{sec:nonextremal})
proves that all non-forbidden, non-extremal graphs are \mbm.

We observe that in order to prove that a graph $G$ is \mbm\
we show that there exists a 
profile of initial preferences \blf\ with a minority of nodes having preference $1$
and a sequence of updates starting from \blf\ that reach a profile 
$\st$ in which at least $n/2$ nodes have preference $1$.
Then, we show that in $\st$
all nodes with preference $1$ are not unhappy, and thus they stay with preference $1$. 
We remark that 
this does not imply that $\st$ is a stable profile but only that in 
any stable profile that is reachable from $\st$ through a sequence of 
updates the nodes with preference $1$ are at least as many as in 
$\st$. This suffices to prove that $G$ is \mbm.

\subsection{Extremal Graphs}
\label{sec:extremal}
We remind the reader that an extremal graph $G$ is a graph 
with an even number $n$ of nodes and in which at least $n-1$
nodes have degree at least $n-2$.

Let $G$ be an extremal graph and let
$u$ be a node of $G$ of minimum degree.
We partition the set $V\setminus\{u\}$ of nodes into $A\cup B\cup C$ where
\begin{enumerate}
\item $A$ is the set of nodes of $V\setminus\{u\}$ that have degree $n-1$ in $G$; 
thus nodes of $A$ are adjacent to all the nodes $V$.
\item $B$ is the set of nodes of $V\setminus\{u\}$ that have degree $n-2$ in $G$ and 
that are adjacent to $u$;
thus nodes of $B$ are adjacent to all the nodes of $V$ except for one.
\item $C$ is the set of nodes of $V\setminus\{u\}$ that have degree $n-2$ 
and that are not adjacent to $u$;
thus nodes of $C$ are adjacent to all nodes in the graph except for $u$.
\end{enumerate}
We will denote by $\alpha$, $\beta$ and $\gamma$ the sizes of $A$, $B$ and $C$, respectively.
Observe that $\deg(u)=\alpha+\beta$.

Nodes of $B$ can be arranged into pairs of non-adjacent nodes. 
Specifically, let $v$ be a node of $B$ and let $w$ be the unique 
node of the graph that is not adjacent to $v$. 
Clearly, $w\not\in A$ (since nodes in $A$ are adjacent to all nodes) and 
$w\not\in C$ (since nodes of $C$ are adjacent to all nodes except for $u$ and $v\ne u$).
Thus each node $v\in B$ is non-adjacent to exactly one node $w$ of $B$
and we call $v$ and $w$ {\em companion} nodes.
In particular, this implies that $B$ has an even number of nodes and 
we name them 
$b_1,b_2,...,b_{2k}$ in such way that
$b_{2i-1}$ and $b_{2i}$ are companion nodes, for $i=1,\ldots,k$,  

\begin{prop}\label{prop:extremal}
Every non-forbidden extremal graph is \mbm.
\end{prop}
\begin{proof}
Let $G$ be a non-forbidden extremal graph with $n$ nodes. 
Consider a profile $\blf$ of initial preferences that satisfies the 
following four properties:
\begin{enumerate}
\item $\blf(v)=1$ for $n/2-1$ nodes;
\item $\blf(u)=0$;
\item exactly $\floor{\deg(u)/2}+1$ nodes adjacent to $u$ have preference $1$;
\item there exists at least one neighbor $x$ of $u$ with $\blf(x)=0$ 
and if $x$ has degree $n-2$ then the one node $y$ 
to which $x$ is not adjacent has preference $\blf(y)=0$.
\end{enumerate}

We next prove that there exists a sequence of updates that,
starting from the truthful profile $\blf$, leads to a stable profile with at least 
$n/2+1$ nodes with preference $1$.
Then we shall prove that, if $G$ is not forbidden, 
a profile $\blf$ of initial preferences satisfying the four properties above always exists.

Let us consider the following sequence of updates.
Since in the truthful profile $\blf$
at least $\floor{\deg(u)/2}+1$ neighbors of $u$ have preference $1$,
$u$ is unhappy and adopts preference $1$. We thus reach a profile
$\st_1$ in which exactly $n/2$ nodes have preference $1$.

Now let us look at node $x$.
If $x$ has degree $n-1$
(that is, $x$ is adjacent to all nodes of the graph),
then $x$ has $n/2$ neighbors with preference $1$ and $n/2-1$ 
neighbors with preference $0$. Therefore, $x$ is unhappy and adopts preference $1$.
Suppose instead that $x$ has degree $n-2$. Then the last property listed above
implies that the $n/2$ nodes with 
preference $1$ are all adjacent to $x$. Therefore, again, $x$ is unhappy and adopts preference $1$.

Now we have reached a profile $\st_2$ in which $n/2+1$ nodes have 
preference $1$
(these are the initial $n/2-1$ nodes with preference $1$ plus $u$ and $x$). 
We show that they are not unhappy and stays with preference $0$.
This is obvious for $u$ and for $x$ and
thus we only need to consider the $n/2-1$ node with initial preference $1$.

Consider a neighbor $v$ of $u$ with degree $n-1$ and $\st_2(v)=1$.
Then $v$ is adjacent to all $n/2$ nodes with preference $1$ 
(all of them except for $v$ itself) and to 
$n/2-1$ nodes with preference $1$. Thus 
$v$ is not unhappy and stays with preference $1$.
Consider a neighbor $v$ of $u$ with degree $n-2$ and $\st_2(v)=1$.
Then $v$ is adjacent to all nodes in the graph except for one;
thus $v$ has at least $n/2-1$ neighbors with preference $1$ and 
at most $n/2-1$ neighbors with preference $0$.
Hence, $v$ is not unhappy and stays with preference $1$.
Finally, let us consider a node $v$ with degree $n-2$ and $\st_2(v)=1$
that is not adjacent to $u$. Then $v$ is adjacent to exactly $n/2-1$ 
neighbors with preference $1$ and $n/2-1$ neighbors with preference $0$.
Hence, $v$ is not unhappy and stays with preference $1$.

\smallskip
Now, we prove that 
if $G$ is not forbidden then a profile of initial preferences satisfying conditions 1-4 
can be constructed.
First of all observe that to satisfy condition $3$ it must be 
that 
$$\floor{\deg(u)/2}+1\leq n/2-1$$
which is always satisfied since the graph is not forbidden 
and thus $\deg(u)\leq n-3$.  Next to verify that condition $4$
can always be satisfied when the graph is not forbidden, 
we consider the following cases:

\begin{description}

\item[$\mathbf{B=\emptyset}$.]

In this case $G\setminus\{u\}$ is a clique with $n-1$ nodes and,
since $G$ is not forbidden, $\alpha=\deg(u)\geq 3$.
Thus we construct $\blf$ by assigning preference $1$ to $\floor{\alpha/2}+1$
nodes from $A$ and to $n/2-\floor{\alpha/2}-2$ nodes from $C$.
We pick $x$ from the  $\ceil{\alpha/2}-1\geq 1$ nodes of $A$ with preference $0$. 

\item[$\mathbf{B\ne\emptyset}$.]

In this case $G\setminus\{u\}$ is not a clique and, 
since $G$ is not forbidden, $\alpha+\beta=\deg(u)\geq 5$.
We distinguish two subcases.

    \begin{description}
            \item[$\mathbf{A\ne\emptyset}$.]
We construct $\blf$ by assigning preference $1$ to $\beta/2+1$ nodes
from $B$ and to $\floor{\alpha/2}$ nodes from $A$.  The remaining nodes
with preference $1$ are taken from $C$. Finally, we pick $x$ from the 
$\ceil{\alpha/2}\geq 1$ nodes of $A$ with preference $0$.
    \item[$\mathbf{A=\emptyset}$.]
    In this case $\beta \geq 6$.
    We construct $\blf$ by assigning preference $1$ to $\beta/2+1$ nodes
    from $B$ and the remaining nodes are taken from $C$.
    Specifically, we assign preference $1$ to nodes $b_1,\ldots,b_{k+1}$
    so that $b_{2k}$ and its companion $b_{2k-1}$ are assigned preference
    $0$. Then we pick $x$ to be $b_{2k}$.
    \end{description}
\end{description}
\end{proof}
We remark that the proof of Proposition~\ref{prop:extremal} immediately gives 
a polynomial-time algorithm for computing an \mbm\ profile of initial preferences for 
non-forbidden extremal graphs.

\subsection{Non-Extremal Graphs}
\label{sec:nonextremal}
In this section we consider {\em non-extremal} graphs. 
These are graphs with an even number $n$ of nodes in which at least
two nodes have degree less than $n-2$.
In this section we prove the following proposition.
\begin{prop}\label{prop:nonextremal}
Every non-forbidden, non-extremal graph is \mbm.
\end{prop}
The proof of Proposition~\ref{prop:nonextremal} relies
on graph bisections that satisfy particular properties. 
The most important among these properties is that the width 
of these bisections is the local minimum with respect to a 
specific neighborhood function between bisections. This 
neighborhood function is computable in polynomial-time:
this gives rise to a polynomial-time local-search algorithm 
that, given a non-forbidden non-extremal graph $G$, computes 
a locally minimum bisection.
However, for sake of readability,
we will ignore this issue throughout this section
and instead present 
a non-constructive proof using (globally) minimum bisections.
We identify three {\em special} types of minimum bisections:
\mtwo, \mone, and \mthree\ 
(see Definitions~\ref{def:mtwo},~\ref{def:mone} and~\ref{def:mthree}).
Then, we prove 
(see Lemmas~\ref{lem:mtwo},~\ref{lem:mone} and~\ref{lem:mthree})
that if a graph $G$ has a minimum bisection that is special, then it is \mbm\ (i.e., there exists an \mbm\ profile of initial preferences for $G$ that is constructed using the special bisection).
Finally, we show that all non-forbidden graphs admit a 
minimum bisection that is special. We do so by 
partitioning the set of minimum bisections into three classes:
{\em weak}, {\em strong}, and 
{\em zero};
then we prove the claim for each of the three classes 
(see Lemmas~\ref{lem:weak},~\ref{lem:strong} and~\ref{lem:zero}).
Finally
we discuss how the assumption
of minimum bisections can be weakened and how local-search can be
used to compute the \mbm\ profile of initial preferences.

\subsubsection{Bisections and Deficiency}
\label{sec:bisection}
We remind the reader that a
{\em bisection} $\S=(S,\nS)$ of a graph $G=(V,E)$ with an even number $n$ 
nodes is simply a partition of the nodes of $V$ into two sets $S$ and $\nS$ 
each of size $n/2$. 
We also recall that, 
given a bisection $\S=(S,\nS)$, we define
the {\em deficiency $\defi_\S(x)$ of node $x$ with respect to 
bisection $\S$} as
$$\defi_\S(x)= \begin{cases}
                W(x,S)-W(x,\nS), & \text{if } x\in S;\\
                W(x,\nS)-W(x,S), & \text{if } x\in \nS.
               \end{cases}
$$

%
For non-extremal graphs the following property of minimum bisections will be useful.
\begin{lemma}\label{lem:strongminimal}
Let $\S=(S,\nS)$ be a bisection of a non-extremal graph $G$. 
If one side of $\S$ has a node $z$ with degree $n-2$ and $\defi_\S(z)=-2$ and 
the other side is a clique with all nodes $x$ having 
$\defi_\S(x)=0$, then $\S$ is not a minimum bisection.
\end{lemma}
\begin{proof}
Without loss of generality assume that $z \in S$ and $\nS$ is a clique. 
Since every $y \in \nS$ has $\defi_\S(y) = 0$, then $y$ has 
$n/2-1$ neighbors in $\nS$ and $n/2-1$ neighbors in $S$. 
Therefore, all nodes in $\nS$ have degree $n-2$ and  
\begin{equation}\label{eq:cutsize}
W(S, \nS) = \frac n 2 \left( \frac n 2 - 1\right).
\end{equation}
In the rest of the proof we will show that it is possible to construct a new bisection $\S'=(S', \nSp)$ such that 
$W(S', \nSp) < W(S, \nS)$.

We partition nodes of $G$ in four sets: $A$ contains nodes of degree less than $n-2$ and,
since $G$ is not extremal, it has size $\alpha \geq 2$;
$B$, whose size is denoted by $\beta$, contains nodes of degree $n-2$
whose companion (their unique non-adjacent node) has degree less than $n-2$ (and thus it belongs to $A$); 
$C$, whose size is denoted by $\gamma$, contains nodes of degree $n-2$ whose companion has degree $n-2$ (and thus it belongs to $C$); 
$D$ contains the nodes of degree greater than $n-2$ and it has size $\delta$.

By the previous observations we have that all nodes in $A$ and 
$D$ belong to $S$. Moreover, since $\nS$ is a clique, a node $x \in C$ and her companion cannot both be in $\nS$ and thus at least half of the nodes in $C$ belong to $S$. Finally, either $z\in B$
or $z\in C$, but in this last case its companion should be in $S$ too, otherwise $\defi_\S(z) = 0$. Thus, since $|S| = n/2$, we have that
\begin{equation}\label{eq:Ssize}
\frac n 2 \geq \alpha+\delta+\frac{\gamma}{2}+1.
\end{equation}

Assume now there exists a node $u\in A$ such that $W(u, B) \geq (n/2-\gamma/2-\delta-1)$
(we will next prove that such a node always exists)
and set $k=|W(u, B)|$.
Consider now the bisection $\S'=(S', \nSp)$ where $S'$ contains node $u$, $\min\{k, n/2 - \gamma/2 -1\}$ nodes of $B$ chosen among the neighbors of $u$, $\gamma/2$ nodes of $C$ chosen in such a way that they form a clique, and $\max\{0, n/2 - \gamma/2 - k - 1\}$ nodes in $D$.
Observe that such a bisection is well defined.
Indeed, by hypothesis, there are enough neighbors of $u$ in $B$,
$C$ contains a clique of $\gamma/2$ nodes
and $\delta \geq n/2 - \gamma/2 - k - 1$ (since $k \geq n/2-\gamma/2-\delta-1$).
Moreover, by construction, $S'$ is a clique and thus each of its nodes has $n/2-1$ neighbors in $S'$.

Let us now compute $W(S', \nSp)$. We distinguish two cases, depending on the value of $k$. 

If $k \geq n/2 - \gamma/2 -1$, then $S'$ contains node $u$,
$n/2 - \gamma/2 -1$ nodes of $B$, $\gamma/2$ nodes of $C$ and no node of $D$.
Since nodes in $B$ and $C$ have $n/2-1$ neighbors in $\nSp$ and $u$ has $\deg(u)-(n/2 -1)$ neighbors in $\nSp$ we have that
$$
W(S', \nSp) = \deg(u) - \left(\frac n 2 - 1\right) + \left(\frac n 2 - 1\right)^2 = \frac n 2 \left(\frac n 2 - 1\right) + \deg(u) - (n-2).
$$
Since $u$ has degree less than $n-2$ then we can conclude that
$W(S', \nSp) \leq \frac n 2 \left(\frac n 2 - 1\right) - 1$ and 
$\S'$ has width strictly smaller than $\S$.

If $k < n/2 - \gamma/2 -1$, then $S'$ contains node $u$, $k$ nodes of $B$, $\gamma/2$ nodes of $C$ and $n/2 - \gamma/2 - k - 1$ nodes of $D$.
Since nodes in $B$ and $C$ have $n/2-1$ neighbors in $\nSp$, nodes in $D$ have $n/2$ neighbors in $\nSp$
and $u$ has $\deg(u)-(n/2 -1)$ neighbors in $\nSp$ we have that
\begin{align*}
W(S', \nSp) & = \deg(u) - \left(\frac n 2 - 1\right) + \left( k+\frac{\gamma}{2} \right) \left(\frac n 2 - 1\right) + \frac n 2 \left( \frac n 2 - \frac{\gamma}{2} - k - 1\right) \\
   & \leq \frac n 2 \left(\frac n 2 - 1\right) + \frac{\gamma}{2} + \delta + \alpha - \frac n 2
\end{align*}
where the inequality holds since $\deg(u) \leq \gamma + k + \delta + \alpha - 1$.
From \eqref{eq:Ssize} we have that $W(S', \nSp) \leq \frac n 2 \left(\frac n 2 - 1\right) - 1$ and 
$\S'$ has width strictly smaller than $\S$.

To conclude the proof it remains to prove that a 
node $u\in A$ such that $W(u, B) \geq (n/2-\gamma/2-\delta-1)$ always exists.
Assume for sake of contradiction, that all nodes in $A$ have less than  $(n/2-\gamma/2-\delta-1)$ neighbors in $B$. 
Then,
$$
W(A, B) < \alpha \left(\frac n 2 - \frac{\gamma}{2} - \delta - 1\right).
$$
On the other hand, by definition, a node in $B$ is adjacent to all the nodes in the 
graph except for her companion that belongs to $A$. Thus, 
$$
W(A, B) = \beta(\alpha-1) = (n-\alpha-\gamma-\delta)(\alpha-1).
$$
Hence, we have that
\begin{equation}
\label{eq:condition_alpha}
 \left(\alpha-1\right) \left(n-\alpha-\gamma-\delta\right) < \alpha \left( \frac n 2 - \frac{\gamma}{2} - \delta - 1\right).
\end{equation}
Let $f(\alpha) = \alpha^2 - \alpha \left( \frac n 2 - \frac{\gamma}{2} + 2 \right) + \left(n - \delta - \gamma \right)$.
By simple algebraic manipulations, we can see that \eqref{eq:condition_alpha} is satisfied
if only if $f(\alpha) > 0$, where $\alpha$ is the size of the set $A$ and thus can only assume values in $\{2, \ldots, n/2-\gamma/2-\delta-1\}$.
We next show that $f(\alpha) \leq 0$ for any admissible $\alpha$, by reaching in this way a contradiction.

Indeed, it is easy to see that the function $f(\alpha)$ is increasing for $\alpha \geq n/4 - \gamma/4 + 1$.
Thus, it has its local maximum in the extremes of the domain.
But, we can easily see that $f(2) = 4 - 2 \left( \frac n 2 - \frac{\gamma}{2} + 2 \right) + \left(n - \delta - \gamma \right) = -\delta \leq 0$
and
\begin{align*}
f\left(\frac n 2 - \frac{\gamma}{2}-\delta-1\right) & = \left(\frac n 2 - \frac{\gamma}{2} - \delta -1\right)^2\\
 & \qquad - \left(\frac n 2 - \frac{\gamma}{2} - \delta-1\right) \left( \frac n 2 - \frac{\gamma}{2} + 2 \right) + \left(n - \delta - \gamma \right) \\
  & = -(\delta+1)\left(\frac n 2 - \frac{\gamma}{2} - \delta - 3\right) - \delta\\
  & = -(\delta+1) (\alpha -2) - \delta \leq 0,
\end{align*}
where the inequality follows since $\alpha \geq 2$ and $\delta \geq 0$.
\end{proof}

\subsubsection{Special Bisections}\label{sec:special}
We define three types of {\em special} minimum bisections and,
for each of them, we prove that it constitutes a 
sufficient condition for a non-forbidden graph to be $\mbm$.

\begin{definition}\label{def:mtwo}
An \mtwo\ bisection $\S=(S,\nS)$ is a minimum bisection in which
one side of the bisection has a node $z$ with $\defi_\S(z)\leq 0$
and the other side has all nodes $x$ with $\defi_\S(x)\geq -1$ and 
two non-adjacent nodes $u$ and $v$ that are both adjacent to $z$. 
\end{definition}

\begin{lemma}\label{lem:mtwo}
Every non-forbidden, non-extremal graph that has an \mtwo\ bisection is \mbm.
\end{lemma}
\begin{proof}
Let $G$ be a non-forbidden, non-extremal graph with $n$ nodes and 
let $\S=(S,\nS)$ be an \mtwo\ bisection of $G$. 
Assume, without loss of generality, that $z\in\nS$ and $u,v\in S$.
Consider the following profile $\blf$ of initial preferences: 
$\blf$ assigns preference 
$1$ to $z$ and to all nodes of $S$ except for $u$ and $v$;
$\blf$ assigns preference $0$ to $u$ and $v$ and to all nodes
of $\nS$ except for $z$.
It is easy to verify that $\blf$ assigns preference $1$ to 
exactly $n/2-1$ nodes.
In the truthful profile $\blf$,
$u$ is unhappy and adopts preference $1$.
In fact, since $u$ is adjacent to $z$ and not adjacent to $v$, 
then it has $W(u,\nS)-1$ neighbors with preference 0 and $W(u,S)+1$
with preference 1. The claim follows from $\defi_\S(u)\geq -1$.
Similarly, $v$ is unhappy and adopts preference $1$.
We have thus reached a profile $\st_1$ in which there are $n/2+1$ nodes
with preference $1$ (the initial $n/2-1$ plus $u$ and $v$)
and $n/2-1$ nodes with preference $0$.
We complete the proof by showing that in $\st_1$
every node with preference $1$ is not unhappy and thus stays with preference $1$. 
This is obvious for $u$ and $v$.
Node $z$ has preference $\st_1(z)=1$ and it is adjacent
to $W(z,S)$ nodes with preference $1$ and $W(z,\nS)$ nodes with preference $0$. 
Since, $\defi_\S(z)\leq 0$, $z$ is not unhappy and stays with preference 1.
Finally, let us consider a generic node $y\in S\setminus\{u, v\}$. 
If $\defi_\S(y)=-1$ then, by Lemma~\ref{lem:minimal-ICALP}, $y$ and $z$ are adjacent and thus
$y$ has $W(y,\nS)-1$ neighbors with 
preference $0$ and $W(y,S)+1$ nodes with preference 1.
Moreover, since $\defi_\S(y)=-1$ we have that
$$W(y,S)+1\geq W(y,\nS) > W(y,\nS)-1.$$
If $\defi_\S(y)\geq 0$, instead, 
$$W(y, S) + 1 \geq W(y, \nS) + 1 > W(y, \nS)-1.$$
In both the cases $y$ is not unhappy and stays with preference 1.
\end{proof}

\begin{definition}\label{def:mone}
An \mone\ bisection $\S=(S,\nS)$ is a minimum bisection
in which one side of the bisection has all nodes $x$ with 
$\defi_\S(x)\geq -1$ and at least one node $u$ with $\defi_\S(u)>0$.
\end{definition}

The following lemma will be very useful in 
proving that non-extremal non-forbidden graphs with an \mone\ bisection 
are \mbm.

\begin{lemma}\label{lem:tone}
Suppose that a non-extremal graph $G$ with $n$ nodes 
admits a bisection $\S=(S,\nS)$ in 
which $S$ consists of nodes with non-negative deficiency and
includes at least one of positive deficiency. Then $G$ is \mbm.
\end{lemma}
\begin{proof}
Let $v$ be the node with positive deficiency in $S$.

Consider now a profile $\blf$ of initial preferences that assigns preference $1$ to all 
nodes of $S$ except for $v$ and preference $0$ to $v$ and to all nodes 
of $\overline{S}$. Observe that in the truthful profile $\blf$,
$v$ is adjacent to $W(v,S)$ nodes with preference $1$ and
to $W(v,\overline{S})$ nodes with preference $0$. 
Since $\defi_{\mathcal{S}}(v)>0$ then $v$ is unhappy and adopts
preference $1$. We thus reach a profile $\st_1$ in which $n/2$ nodes
have preference $1$ (all nodes in $S$).
We conclude the proof of the lemma by showing that 
every node of $S$ is not unhappy and stays with preference $1$.
This is obvious for $v$. Let us consider $u\in S$ and $u\ne v$.
Then $u$ has $W(u,S)$ neighbors with preference $1$ and 
$W(u,\overline{S})$ neighbors with preference $0$. 
Since $\defi_{\mathcal{S}}(u)\geq 0$, we have that $W(u,S)\geq W(u,\overline{S})$
which implies that $u$ is not unhappy and stays with preference $1$.
\end{proof}
We remark that Lemma~\ref{lem:tone} does not require the bisection $\mathcal{S}$ of the claim to be a minimum bisection.

\begin{lemma}\label{lem:mone}
Every non-forbidden, non-extremal graph that has an \mone\ bisection is \mbm.
\end{lemma}
\begin{proof}
Let $G$ be a non-forbidden, non-extremal graph with $n$ nodes and denote by $\S=(S,\nS)$ the 
\mone\ bisection of $G$. 
Assume, without loss of generality, 
that all  $x\in S$ have $\defi_\S(x)\geq -1$
and that there exists $u\in S$ with $\defi_\S(u)>0$.
If all $x\in S$ have $\defi_\S(x)\geq 0$ then the claim follows
by Lemma~\ref{lem:tone}.
Similarly, 
if all $y\in\nS$ have $\defi_\S(y)\geq 0$ and 
there exists $z\in\nS$ with $\defi_\S(z)>0$ then the claim follows 
from Lemma~\ref{lem:tone} (when applied to $(\nS, S)$).

Suppose that there exists $v\in S$ with $\defi_\S(v)=-1$ and 
$z\in\nS$ with $\defi_\S(z)<0$. 
Clearly, by Lemma~\ref{lem:minimal-ICALP}, it must
be that $\defi_\S(z)=-1$ and that $v$ and $z$ are adjacent. 
Consider profile $\blf$ of initial preferences that assigns preference $1$ to all 
nodes of $S$ except for $u$ and preference $0$ to $u$ and to all 
nodes of $\nS$.
Now observe that, in the truthful profile $\st_0$,
$u$ is adjacent to $W(u,S)$ nodes with preference $1$ and 
to $W(u,\nS)$ nodes with preference $0$. Since $\defi_\S(u)>0$,
it follows that $W(u,S)>W(u,\nS)$. Thus $u$ is unhappy and adopts preference $1$.
As a result of this update, 
we reach a profile $\st_1$ in which there are exactly $n/2$ nodes with 
preference $1$ (all of $S$) and $n/2$ with preference $1$ (all of $\nS$).
In profile $\st_1$, $z$ is adjacent to $W(z,S)$ nodes with preference $1$
and to $W(z,\nS)=W(z,S)-1$ nodes with preference $0$. Therefore
$z$ is unhappy and adopts preference $1$. 
We thus reach profile $\st_2$ in which $n/2+1$ nodes have preference $1$
(these are the $n/2-1$ nodes with preference $1$ plus $u$ and $z$)
and $n/2-1$ have preference $0$.
Clearly, 
in profile $\st_2$, $u$ and $z$ are unhappy and stay with preference $1$. 
Consider now node $x\in S$ with $x\ne u$.
Node $x$ is adjacent to $W(x,S)+W(x,z)$ nodes with preference $1$ and to 
$W(x,\nS)-W(x,z)$ nodes with preference $0$. 
If $\defi_\S(x)\geq 0$ then $W(x,S)\geq W(x,\nS)$ and thus $x$ is not unhappy
and stays with preference $1$.
If instead $\defi_\S(x)=-1$ then, by Lemma~\ref{lem:minimal-ICALP},
$x$ and $z$ are adjacent and thus 
$W(x,S)+W(x,z)=W(x,\nS)$ 
and, again, $x$ is not unhappy and stays with preference $1$.
The claim thus follows.

Let us now consider the case in which 
there exists $v\in S$ with $\defi_\S(v)=-1$ and 
all $y\in\nS$ have $\defi_\S(y)=0$. 
By Lemma~\ref{lem:minimal-ICALP}, $v$ is adjacent to all nodes in $\overline S$. 
Therefore if $\nS$ is not a clique, $\S$ is an \mtwo\ bisection and 
thus the claim follows from Lemma~\ref{lem:mtwo}.
If instead $\nS$ is a clique then for all $y\in\nS$ we have that 
$W(y,\nS)=n/2-1$ and, since $\defi_\S(y)=0$, 
$W(y,S)=n/2-1$. This implies that all $y\in\nS$ have $\deg(y)=n-2$ and
that $W(S,\nS)=n/2(n/2-1)$. 
If $u$ has no neighbor in $\nS$, 
the fact that 
$W(S,\nS)=n/2(n/2-1)$ implies that each of the remaining $n/2-1$ nodes
$x\in S\setminus\{u\}$ have $W(x,\nS)=n/2$. Since $\defi_\S(x)\geq -1$
we have that $W(x,S)\geq n/2-1$ which implies that $\deg(x)=n-1$.
We can then conclude that the graph $G$ is extremal thus reaching a 
contradiction.

Thus the last case left 
is when there exists $v\in S$ with $\defi_\S(v)=-1$,
$\nS$ is a clique,
all $y\in\nS$ have $\defi_\S(y) = 0$ and $u$ has a neighbor $z\in\nS$.
For this case, 
we consider a different profile $\blfp$ of initial preferences:
$\blfp$ assigns preference $1$
to $z$ and to all nodes of $S$ except for $u$ and $v$;
nodes $u$ and $v$ and all nodes of $\nS$ except for $z$ have preference $0$.
In the truthful profile $\blfp$,
$u$ is adjacent to at least $W(u,S)+1-W(u,v)\geq W(u,S)$ 
nodes with preference $1$ and to at most $W(u,\nS)-1+W(u,v)\leq W(u,\nS)$ 
nodes with preference $0$.
Since $\defi_S(u)>0$, we have that $W(u,S)>W(u,\nS)$ and thus
$u$ is unhappy and adopts preference $1$.
We thus reach a profile $\st_1^\prime$ with the same number of nodes with 
preference $0$  and preference $1$.
Node $v$ has preference $\st_1^\prime(v)=0$ and, in profile $\st_1^\prime$, 
$v$ is adjacent to 
$W(v,S)+1=W(v,\nS)$ nodes with preference $1$ 
(remember that, by Lemma~\ref{lem:minimal-ICALP}, $z$ and $v$ are adjacent)
and to  $W(v,\nS)-1$ nodes with preference $0$.
Therefore $v$ is unhappy in $\st_1^\prime$ and adopts preference $1$.
We thus reach a profile $\st_2^\prime$ with $n/2+1$ nodes with 
preference $1$ and $n/2-1$ nodes with preference $0$.
We conclude the proof by showing that in $\st_2^\prime$
all nodes with preference $1$ are not unhappy and stay with preference $1$.
This is obvious for $v$ and $u$.
Consider now a node $x\in S$ other than $u$ and $v$.
In $\st_2^\prime$ 
node $x$ has $W(x,S)+W(x,z)$ neighbors with preference $1$ and 
$W(x,\overline{S})-W(x,z)$ nodes with preference $0$.
If $\defi_\S(x)\geq 0$ then $x$ is not unhappy and stays with 
preference $1$. If instead $\defi_\S(x)=-1$ then, by Lemma~\ref{lem:minimal-ICALP},
$x$ and $z$ are adjacent and thus also in this case $x$ is not unhappy
and stays with preference $1$.
Finally, in $\st_2^\prime$ node $z$ is adjacent to $W(z,S)$ nodes with 
preference $1$ and to $W(z,\nS)$ nodes with preference $0$.
Since $\defi_\S(z)=0$, the claim follows.
\end{proof}

\begin{definition}\label{def:mthree}
An \mthree\  bisection $\S=(S,\nS)$ is a minimum bisection in which
all nodes $x$ of one side have $\defi_\S(x)=0$ and 
all nodes $y$ of the other side have $\defi_\S(y) \in \{-1,0\}$
and at least one node $u$ has $\defi_\S(u)=-1$.
\end{definition}

\begin{lemma}\label{lem:mthree}
Every non-forbidden, non-extremal graph that has an \mthree\ bisection is \mbm.
\end{lemma}
\begin{proof}
Let $G$ be a non-forbidden, non-extremal graph with $n$ nodes and denote by 
$\S=(S,\nS)$ the \mthree\ bisection of $G$.
Assume, without loss of generality, 
that all nodes of $S$ have zero deficiency, 
that all nodes of $\nS$ have deficiency at least $-1$,
and that $u\in\nS$. Then, by Lemma~\ref{lem:minimal-ICALP}, 
$u$ is adjacent to all nodes of $S$.
If $S$ is not a clique then $\S$ is an \mtwo\ bisection and the claim
follows from Lemma~\ref{lem:mtwo}.

Suppose then that $S$ is a clique. Therefore, for all $x\in S$,  
$W(x,S)=n/2-1$ and, since $\defi_\S(x)=0$, we have that 
$\deg(x)=n-2$.
We observe that nodes with zero deficiency have even degree.
Moreover, every node $y\in\nS$ with $\defi_\S(y)=-1$ (including $u$) is, 
by Lemma~\ref{lem:minimal-ICALP}, adjacent to all nodes of $S$ and thus
$\deg(y)=n-1$ (which is odd). 
This observation has two important consequences. 
First, since the number of odd-degree nodes in a graph is even, 
there must be at least one node $v\ne u$ with odd degree and 
it must be the case that $v\in\nS$ and that $v$ has $\defi_\S(v)=-1$ and degree $n-1$.
Second, since $G$ is non-forbidden, there must exists a node $z$
with $\deg(z)\leq n-3$ and it must be the case that 
$z\in\nS$ and $\defi_\S(z)=0$. 
In turn this implies that there exist $x_1,x_2\in S$ that are not 
adjacent to $z$. Notice that, since $u$ has degree $n-1$, 
$z$ is adjacent to $u$.  
Consider now the following profile $\blf$ of initial preferences:
$\blf$ assigns preference $1$ to $u$ and to all nodes of $S$ except for
$x_1$ and $x_2$;
$\blf$ assigns preference $0$ to $x_1$ and $x_2$ and to all nodes of 
$\nS$ except for $u$.
It is easy to verify that $\blf$ assigns preference $1$ to exactly $n/2-1$ nodes.
In the truthful profile $\blf$,
$z$ is unhappy and switches from preference $0$ to preference $1$.
In fact, observe that 
all nodes of $\nS$ adjacent to $z$ except for $u$ have 
preference $0$ and $z$ is not adjacent to the two nodes of $S$, 
$x_1$ and $x_2$, that have preference $0$.
Therefore, $z$ is adjacent to $W(z,\nS)-1$ nodes with preference $0$.
On the other hand, all nodes of $S$ adjacent to $z$ have preference $1$ and
$z$ is adjacent to $u$. 
Therefore, $z$ is adjacent to $W(z,S)+1$ nodes with preference $1$.
The claim then follows from $\defi_\S(z)=0$.
We have thus reached a profile $\st_1$ in which the number of nodes
with preference $0$ and preference $1$ are equal.

In $\st_1$,
$v$ is unhappy and switches from preference $0$ to preference $1$.
This follows from the fact that $v$ is adjacent to all nodes of the graph
(it has degree $n-1$) and thus
$v$ is adjacent to $n/2-1$ nodes with preference $0$ and to 
$n/2$ nodes with preference $1$.

We have thus reached a profile $\st_2$ in which there are $n/2+1$ nodes
with preference $1$ (the initial $n/2-1$ plus $z$ and $v$)
and $n/2-1$ nodes with preference $0$.
We complete the proof by showing that in $\st_2$
every node with preference $1$ is not unhappy and stays with preference $1$. 
This is obvious for $v$ and $z$.
Node $u$ has $\deg(u)=n-1$ and it is adjacent
to $n/2$ nodes with preference $1$ and $n/2-1$ with preference $0$.
Thus, $u$ is not unhappy and stays with preference $1$.
Let us now consider a generic node $y\in S$. Node $y$ has degree $n-2$ and 
thus it is adjacent to all
nodes of the graph except for one.
Therefore $y$ has at least $n/2-1$ adjacent nodes with preference $1$ 
and at most $n/2-1$ nodes with preference $0$.
Hence, $y$ is not unhappy and stays with preference $1$.
\end{proof}

\subsubsection{Proof of Proposition~\ref{prop:nonextremal}}\label{sec:proof-nonextremal}
We partition the set of minimum bisections of a non-extremal
graph $G$ into three classes:
{\em weak}, {\em strong} and {\em zero}.
Specifically, let $\S$ be a minimum bisection. 
We call $\S$ {\em weak} if at least one node $z$ has 
deficiency $\defi_\S(z)<-1$.
$\S$ is called {\em strong} if all nodes have deficiency 
at least $-1$ and at least one has deficiency different from $0$.
Finally, if all nodes have deficiency $0$, $\S$ is called {\em zero}.

\begin{lemma}\label{lem:weak}
Every non-forbidden, non-extremal graph that admits a weak bisection is \mbm.
\end{lemma}
\begin{proof}
Let $G$ be a non-forbidden, non-extremal graph with $n$ nodes that admits a weak bisection $\S=(S,\nS)$. Let us consider first the case in which there exists $z$ with 
$\defi_\S(z)\leq -3$ and suppose, 
without loss of generality, that $z\in\nS$.
Then, by Lemma~\ref{lem:minimal-ICALP}, for all $x\in S$, 
$\defi_\S(x)\geq 1$ and thus $\S$ is \mone. 
The claim follows by Lemma~\ref{lem:mone}.

Suppose now that all nodes have deficiency at least $-2$ 
and that there exists a node $z$ with $\defi_\S(z)=-2$.  
Again, without loss of generality, assume that $z\in\nS$.
Then, by Lemma~\ref{lem:minimal-ICALP}, for all $x\in S$, 
$\defi_\S(x)\geq 0$. 
If there exists $u\in S$ with $\defi_\S(u)>0$ then $\S$ is \mone.
The claim then follows by Lemma~\ref{lem:mone}.

Suppose now that all $x\in S$ have $\defi_\S(x)=0$. 
Since $\S$ is a minimum bisection, by Lemma~\ref{lem:strongminimal}, $S$ is not a
clique and, by Lemma~\ref{lem:minimal-ICALP}, $z$ is adjacent to all 
nodes of $S$. But then $\S$ is \mtwo\ and the claim follows from 
Lemma~\ref{lem:mtwo}.
\end{proof}

The next lemma deals with strong bisections.
\begin{lemma}\label{lem:strong}
Every non-forbidden, non-extremal graph that admits a strong bisection is \mbm.
\end{lemma}
\begin{proof}
Let $G$ be a non-forbidden, non-extremal graph with $n$ nodes 
that admits a strong bisection $\S=(S,\nS)$. If there exists 
node $u$ with $\defi_\S(u)>0$ then $\S$ is \mone\ and the
claim follows by Lemma~\ref{lem:mone}.
Let us consider then the case in which $\S$ is not \mone\ and thus 
all nodes $x$ have $\defi_\S(x)\in\{-1,0\}$.
We partition $S$ into $S_0$ containing any $x \in S$ with $\defi_\S(x)=0$
and $S_1$ containing any $x \in S$ with $\defi_\S(x)=1$.
Similarly, we partition $\nS$ into $\nS_0$ and $\nS_{-1}$.
Note that at least one of 
$S_{-1}$ and $\nS_{-1}$ is non-empty.

We fist observe that if one among $S_0$ and $\nS_0$ is non-empty then the
other is also non-empty.
Suppose in fact that $S_0\ne\emptyset$ and 
assume, for sake of contradiction, that $\nS_0=\emptyset$ 
(and thus all nodes $y$ of $\nS$ have $\defi_\S(y)=-1$). 
Then, by Lemma~\ref{lem:minimal-ICALP}, we have that,
for every $x\in S_0$, $W(x,\nS)=n/2$ and, 
since $\defi_\S(x)=0$, it must be the case that $W(x,S)=n/2$. 
This is impossible. 

Let us now consider the case in which $S_0,\nS_0=\emptyset$ or, 
equivalently, that every node $x\in S\cup\nS$ has $\defi_\S(x)=-1$.
Lemma~\ref{lem:minimal-ICALP} then implies that $x$ is adjacent to all
nodes on the opposite side and this, together with $\defi_\S(x)=-1$, 
implies that $x$ is adjacent to all node on its side. 
In other words, $G$ is a clique and thus \ftwo\ and we reached a contradiction.

Let us now consider the case in which $S_0,\nS_0\ne\emptyset$.
Henceforth, we assume without loss of generality that $S_{-1}\ne\emptyset$.
If $\nS$ is not a clique then $\S$ is an 
\mtwo\ bisection and the claim follows by Lemma~\ref{lem:mtwo}.
Similarly, if also $\nS_{-1}\ne\emptyset$ and $S$ is not a clique.

Suppose now that neither $S_{-1}$ nor $\nS_{-1}$ is empty and 
both $S$ and $\nS$ are cliques. Then every node $y$ is adjacent to 
all $n/2-1$ nodes on its side and, since $\defi_\S(y)\geq -1$,
each node has degree at least $n-2$. Therefore the graph is \ftwo\ and 
we reached a contradiction. 
Finally, we are left with the case in which $S_0,S_{-1}\ne\emptyset$
and $\nS_{-1}=\emptyset$ and thus $\nS_0\ne\emptyset$. 
But then $\S$ is an \mthree\ partition and thus the claim follows
by Lemma~\ref{lem:mthree}.
\end{proof}

Before dealing with zero bisections, 
we prove the following technical lemma.
\begin{lemma}\label{lem:zerotech}
Let $\S=(S,\nS)$ be a zero bisection of graph $G$ and let 
$u\in S$ and $z\in\nS$ be two non-adjacent nodes.
Then either bisection ${\cal T}$ obtained from $\S$ when $u$ and $z$ 
switch sides is non-zero (that is, strong or weak) or 
$u$ and $z$ have the same neighborhood.
\end{lemma}
\begin{proof}
As $\defi_\S(u)=\defi_\S(z)=0$ and $W(u,z)=0$, 
by Lemma~\ref{lem:minimal-ICALP}, ${\cal T}$ is a minimum bisection. 
Suppose that there exists $w$ that is adjacent to $u$ and not to $z$. 
Then we have that
\begin{eqnarray*}
W(w,T) & = & W(w,S)-1 \\
W(w,\overline{T}) &=& W(w,\nS)+1
\end{eqnarray*}
whence we obtain $\defi_{{\cal T}}(w)=\pm 2$ and thus ${\cal T}$ is a
non-zero bisection.
The case in which $w$ is adjacent to $z$ but not to $u$ is similar.
\end{proof}

\begin{lemma}\label{lem:zero}
Every non-forbidden, non-extremal graph 
that admits a zero bisection is \mbm.
\end{lemma}
\begin{proof}
Let $G$ be a non-forbidden, non-extremal graph with $n$ nodes that admits a zero bisection $\S=(S,\nS)$. 
We observe that if one side of $\S$, say $S$, is a clique then all 
$x\in S$ have $W(x,S)=n/2-1$ which, together with $\defi_\S(x)=0$, implies
that $\deg(x)=n-2$.  Therefore if both sides are cliques, 
$G$ is \ftwo\ and the claim follows.
Suppose then that $S$ is not a clique. 

Let $v$ be a node in $S$ with $\deg(v)>0$. 
If no such node exists, then the graph is \fone,
because any node has deficiency $0$.
Since $S$ is not a clique there exists $u\in S$ such that 
$u$ and $v$ are not adjacent. 
Moreover, since $\deg(v)>0$ and since $\defi_\S(v)=0$ there exists
$z\in\nS$ such that $v$ and $z$ are adjacent.

Suppose that $u$ and $z$ are adjacent. 
Then, since $u$ and $v$ are not adjacent and $z$ is adjacent to 
both, $\S$ is \mtwo\ and the claim
follows from Lemma~\ref{lem:mtwo}.
If instead $u$ and $z$ are not adjacent, then
they have different neighborhoods  ($z$ is adjacent to $v$ whereas
$u$ is not) and thus, by Lemma~\ref{lem:zerotech},
the bisection ${\cal T}$ obtained when $u$ and $z$ switch sides is non-zero. 
The claim then follows from Lemma~\ref{lem:weak} and Lemma~\ref{lem:strong}.
\end{proof}

\subsubsection{Making the Proof of Proposition \ref{prop:nonextremal} Constructive}\label{sec:algo}
We now explain how we can transform the proof of the theorem into an algorithm. All we have to do is to explain how the assumption of (globally) minimum bisection can be replaced by a property that is testable in polynomial time.

First, by carefully inspecting the proof of Lemma \ref{lem:minimal-ICALP}, we observe that the assumption we need every time we invoke this lemma for some bisection is actually that 
the width of the bisection cannot be improved by swapping two nodes from different sides of the bisection; let us use the term {\em locally minimal} for such a bisection. Clearly, testing whether a bisection is locally minimal can be done in polynomial-time by considering all pairs of swapping nodes.

Hence, we can relax the term ``minimum'' in almost all statements, definitions, and proofs to locally minimal.
The only case where this replacement is not enough is in the definition of weak bisections and in statements which are proved using Lemma \ref{lem:strongminimal}; there, we require that the bisection is not only locally minimal but also does not have the structure assumed in the statement of Lemma \ref{lem:strongminimal}. Let us call such bisections {\em strongly locally minimal}. Here, we would like to be able to detect bisections that have the structure assumed in the statement of Lemma \ref{lem:strongminimal} (this is easy) and furthermore compute new bisections that are strongly locally minimal. Indeed, the proof of Lemma \ref{lem:strongminimal} follows by constructing another bisection that has strictly smaller width starting from a bisection that satisfies the particular conditions of the lemma. This construction is much more complicated
than just switching two nodes but can still be computed in polynomial-time.

Finally, the proof of Lemma \ref{lem:zero} follows by the fact that either the zero bisection is \mtwo\ or there is another bisection (obtained by swapping two non-adjacent nodes from different sides of the zero bisection) that is either weak or strong.
Hence, it is also necessary that the strongly minimal bisection enjoys this property.

So, at a high-level, our local-search algorithm takes as input a non-forbidden non-extremal graph and works as follows. 
\begin{enumerate}
\item It starts from an arbitrary bisection of the input graph.
\item It continuously considers a new bisection of strictly smaller width by swapping two nodes in different sides of the current bisection. When this is not possible anymore, it proceeds with Step~3.
\item It checks whether the current bisection has a node with degree $n-2$ and deficiency $-2$ on one side and the other side is a clique of nodes with deficiency $0$ (i.e., it checks whether the condition of Lemma \ref{lem:strongminimal} apply). If this is the case, it uses the modification in the proof of Lemma \ref{lem:strongminimal} to obtain another bisection of strictly smaller width and goes to Step~2.
\item It checks whether the current bisection is a zero bisection. If this is the case and the bisection is not \mtwo, it computes a bisection by swapping two non-adjacent nodes in different sides of the current bisection and goes to Step~2.
\item It defines an \mbm\ profile of initial preferences for the input graph using the machinery in the proof of 
\begin{itemize}
\item Lemma \ref{lem:mtwo}, if the bisection is \mtwo; 
\item Lemma \ref{lem:mone}, if the bisection is \mone;
\item Lemma \ref{lem:mthree}, if the bisection is \mthree.
\end{itemize}
\item It outputs the \mbm\ profile of initial preferences.
\end{enumerate}

First, observe that the width of the bisection strictly decreases every time a new bisection is constructed at Steps 2 or 3.
Also, observe that if we reach Step 4 and the current bisection is an \mtwo\ zero bisection, the algorithm moves to Steps 5 and 6 and then terminates. Otherwise, suppose that Step 4 constructs a new bisection and the algorithm proceeds with Step 2. Then, by the proof of Lemma \ref{lem:zero}, we have that the new bisection obtained is not zero and has the same width as the previous one. Therefore, if the algorithm reaches Step 4 again with the same bisection, it will proceed to Steps 5 and 6 and terminate. This implies that the width of the bisection strictly decreases every two iterations of Steps 2, 3, and 4 (except possibly the last two), and that the algorithm terminates after $O(n^2)$ steps.

Finally, we remark that when the algorithm reaches Step 5, the current bisection is strongly locally minimal and is either non-zero or zero and \mtwo. Therefore, the proofs of Lemmas \ref{lem:weak}, \ref{lem:strong}, \ref{lem:zero} (can be adapted to) show that such a bisection is either \mtwo, \mone, \mthree. Hence, by applying Lemmas \ref{lem:mtwo}, \ref{lem:mone}, or \ref{lem:mthree}, the desired \mbm\ profile of initial preferences will be computed.

\section{Hardness for Weaker Minorities}
\label{sec:hardness}
We next show that deciding if it is possible to subvert the majority
starting from a weaker minority is a computationally hard problem.
\begin{theorem}
\label{thm:reduction}
For every constant $0 < \varepsilon < \frac{1}{8}$, given a graph $G$ with $n$ nodes,
it is NP-hard to decide whether there exists an \mbm\ profile of initial preferences 
with at most $n\left(\frac{1}{4} - \varepsilon\right)$ nodes with initial preference $1$.
\end{theorem}
%

We will use a reduction from the NP-hard problem 2P2N-3SAT, the problem of deciding whether a 3SAT formula in which every variable appears as positive in two clauses and as negative in two clauses has a truthful assignment or not (the NP-hardness follows by the results of \cite{Y05TRA}).

Given a Boolean formula $\phi$  with $C$ clauses and $V$ variables that is 
an instance of 2P2N-3SAT (thus $3C=4V$ and $C$ is a multiple of $4$), 
we will construct a graph $G(\phi)$ with $n$ nodes such that there exists a 
profile of initial preferences with at most $n\left(\frac{1}{4} - \varepsilon\right)$ nodes of $G(\phi)$ with
preference $1$ such that a sequence of updates can lead to a stable profile in
which at least $n/2$ nodes have preference $1$ if and only if $\phi$ has a satisfying assignment.

The graph $G(\phi)$ has the following nodes and edges.
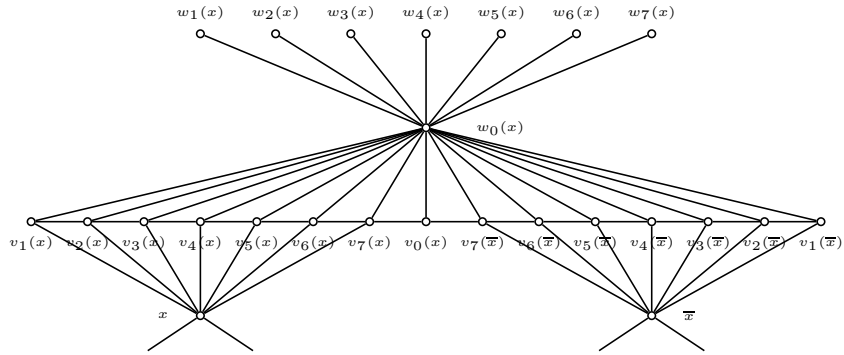
\begin{figure}[htb]
 \begin{tikzpicture}[-,auto,on grid=true,semithick,
                     prof/.style={shape=circle,draw,inner sep=0pt,minimum size=1mm},
                     every label/.style={font=\tiny}]

  \node[prof] (V1) [label={below:$v_{1}(x)$}] {};
  \node[prof] (V2) [right=0.75cm of V1] [label={below:$v_{2}(x)$}] {};
  \node[prof] (V3) [right=0.75cm of V2] [label={below:$v_{3}(x)$}] {};
  \node[prof] (V4) [right=0.75cm of V3] [label={below:$v_{4}(x)$}] {};
  \node[prof] (V5) [right=0.75cm of V4] [label={below:$v_{5}(x)$}] {};
  \node[prof] (V6) [right=0.75cm of V5] [label={below:$v_{6}(x)$}] {};
  \node[prof] (V7) [right=0.75cm of V6] [label={below:$v_{7}(x)$}] {};
  \node[prof] (V0) [right=0.75cm of V7] [label={below:$v_{0}(x)$}] {};
  \node[prof] (V7N) [right=0.75cm of V0] [label={below:$v_{7}(\ox)$}] {};
  \node[prof] (V6N) [right=0.75cm of V7N] [label={below:$v_{6}(\ox)$}] {};
  \node[prof] (V5N) [right=0.75cm of V6N] [label={below:$v_{5}(\ox)$}] {};
  \node[prof] (V4N) [right=0.75cm of V5N] [label={below:$v_{4}(\ox)$}] {};
  \node[prof] (V3N) [right=0.75cm of V4N] [label={below:$v_{3}(\ox)$}] {};
  \node[prof] (V2N) [right=0.75cm of V3N] [label={below:$v_{2}(\ox)$}] {};
  \node[prof] (V1N) [right=0.75cm of V2N] [label={below:$v_{1}(\ox)$}] {};

  \node[prof] (W0) [above=1.25cm of V0] [label={right:$\qquad w_{0}(x)$}] {};

  \node[prof] (W4) [above=1.25cm of W0] [label=$w_{4}(x)$] {};
  \node[prof] (W3) [left=1cm of W4] [label=$w_{3}(x)$] {};
  \node[prof] (W2) [left=1cm of W3] [label=$w_{2}(x)$] {};
  \node[prof] (W1) [left=1cm of W2] [label=$w_{1}(x)$] {};
  \node[prof] (W5) [right=1cm of W4] [label=$w_{5}(x)$] {};
  \node[prof] (W6) [right=1cm of W5] [label=$w_{6}(x)$] {};
  \node[prof] (W7) [right=1cm of W6] [label=$w_{7}(x)$] {};

  \node[prof] (X) [below=1.25cm of V4] [label={left:$x \quad$}] {};
  \node[prof] (XN) [below=1.25cm of V4N] [label={right:$\quad \ox$}] {};

  \phantom{\node[prof] (j) [below=1.75cm of V3] {};}
  \phantom{\node[prof] (k) [below=1.75cm of V5] {};}
  \phantom{\node[prof] (l) [below=1.75cm of V3N] {};}
  \phantom{\node[prof] (m) [below=1.75cm of V5N] {};}

  \draw (X) -- (j);
  \draw (X) -- (k);
  \draw (XN) -- (l);
  \draw (XN) -- (m);
  
  \draw (V1) -- (V2);
  \draw (V2) -- (V3);
  \draw (V3) -- (V4);
  \draw (V4) -- (V5);
  \draw (V5) -- (V6);
  \draw (V6) -- (V7);
  \draw (V7) -- (V0);
  \draw (V0) -- (V7N);
  \draw (V7N) -- (V6N);
  \draw (V6N) -- (V5N);
  \draw (V5N) -- (V4N);
  \draw (V4N) -- (V3N);
  \draw (V3N) -- (V2N);
  \draw (V2N) -- (V1N);
  
  \draw (W0) -- (W1);
  \draw (W0) -- (W2);
  \draw (W0) -- (W3);
  \draw (W0) -- (W4);
  \draw (W0) -- (W5);
  \draw (W0) -- (W6);
  \draw (W0) -- (W7);
  \draw (W0) -- (V1);
  \draw (W0) -- (V2);
  \draw (W0) -- (V3);
  \draw (W0) -- (V4);
  \draw (W0) -- (V5);
  \draw (W0) -- (V6);
  \draw (W0) -- (V7);
  \draw (W0) -- (V0);
  \draw (W0) -- (V7N);
  \draw (W0) -- (V6N);
  \draw (W0) -- (V5N);
  \draw (W0) -- (V4N);
  \draw (W0) -- (V3N);
  \draw (W0) -- (V2N);
  \draw (W0) -- (V1N);
  
  \draw (X) -- (V1);
  \draw (X) -- (V2);
  \draw (X) -- (V3);
  \draw (X) -- (V4);
  \draw (X) -- (V5);
  \draw (X) -- (V6);
  \draw (X) -- (V7);
  \draw (XN) -- (V1N);
  \draw (XN) -- (V2N);
  \draw (XN) -- (V3N);
  \draw (XN) -- (V4N);
  \draw (XN) -- (V5N);
  \draw (XN) -- (V6N);
  \draw (XN) -- (V7N);
\end{tikzpicture}
\centering
\caption{The variable gadget.}
\label{fig:variable}
\end{figure}
For each variable $x$ of $\phi$,
$G(\phi)$ includes a {\em variable} gadget for $x$ consisting of $25$ 
nodes and $50$ edges (see Figure~\ref{fig:variable}).
The nodes of the variable gadget for $x$ are
the {\em literal nodes}, $x$ and $\ox$,
nodes $v_1(x),\ldots,v_{7}(x)$,
nodes $v_1(\ox),\ldots,v_{7}(\ox)$,
nodes $v_0(x)$ and $w_0(x)$, and
nodes $w_1(x), \ldots, w_{7}(x)$.
The edges are
$(x,v_i(x))$ and $(\overline{x},v_i(\overline{x}))$ for $i=1, \ldots, 7$,
$(v_i(x),v_{i+1}(x))$ and $(v_i(\overline{x}),v_{i+1}(\overline{x}))$ for $i=1, \ldots, 6$,
$(v_{0}(x),v_7(x))$, $(v_{0}(x),v_7(\overline{x}))$, $(v_0(x), w_0(x))$,
$(w_0(x),v_i(x))$, $(w_0(x),v_i(\overline{x}))$ and $(w_0(x), w_i(x))$ for $i=1, \ldots, 7$.
\begin{figure}[htb]
 \begin{tikzpicture}[-,auto,on grid=true,semithick,
                     prof/.style={shape=circle,draw,inner sep=0pt,minimum size=1mm},
                     every label/.style={font=\tiny}]

  \node[prof] (V1) [label={left:$\upsilon_1(c)$}] {};
  \node[prof] (V2) [right=1cm of V1] [label={below:$\upsilon_2(c)$}] {};
  \node[prof] (V3) [right=1cm of V2] [label={below:$\upsilon_3(c)$}] {};

  \phantom{\node[prof] (LL) [right=0.5cm of V3] {};}
  \phantom{\node[prof] (RR) [right=5.5cm of V1] {};}
  
  \phantom{\node[prof] (G) [right=4cm of V1] {};}
  \node[prof] (V13) [right=6cm of V1] [label={below:$\upsilon_{13}(c)$}] {};
  \node[prof] (V14) [right=1cm of V13] [label={below:$\upsilon_{14}(c)$}] {};
  \node[prof] (V15) [right=1cm of V14] [label={right:$\upsilon_{15}(c)$}] {};
  
  \phantom{\node[prof] (H) [above=1.25cm of G] {};}
  \node[prof] (U1) [left=0.75cm of H] [label={left:$u_1(c) \qquad$}] {};
  \node[prof] (U2) [right=0.75cm of H] [label={right:$\qquad u_2(c)$}] {};
  
  \node[prof] (C) [above=1.25cm of H] [label={right:$c$}] {};

  \phantom{\node[prof] (L) [above=0.5cm of C] {};}
  \phantom{\node[prof] (M) [left=1cm of L] {};}
  \phantom{\node[prof] (N) [right=1cm of L] {};}

  \draw[dotted] (LL) -- (RR);
  \draw (C) -- (L);
  \draw (C) -- (M);
  \draw (C) -- (N);
  
  \draw (C) -- (U1);
  \draw (C) -- (U2);
  \draw (U1) -- (V1);
  \draw (U1) -- (V2);
  \draw (U1) -- (V3);
  \draw (U1) -- (V13);
  \draw (U1) -- (V14);
  \draw (U1) -- (V15);
  \draw (U2) -- (V1);
  \draw (U2) -- (V2);
  \draw (U2) -- (V3);
  \draw (U2) -- (V13);
  \draw (U2) -- (V14);
  \draw (U2) -- (V15);
\end{tikzpicture}
\centering
\caption{The clause gadget.}
\label{fig:clause}
\end{figure}
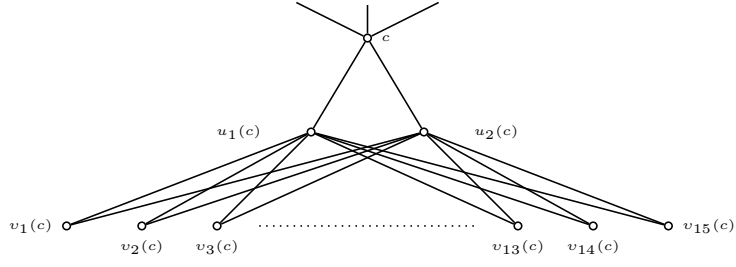
For each clause $c$ of $\phi$, graph $G(\phi)$ includes
a {\em clause gadget} for $c$ consisting of $18$ nodes
and $32$ edges (see Figure~\ref{fig:clause}).
The nodes of the gadget are the {\em clause} node
$c$, nodes $u_1(c)$, $u_2(c)$, and nodes $\upsilon_1(c),\ldots,\upsilon_{15}(c)$.
The $32$ edges are
$(c,u_1(x)), (c,u_2(x))$, and $(u_i(c),\upsilon_j(c))$ with $i=1, 2$ and $j=1, \ldots, 15$.
%
In $G(\phi)$, for every clause $c$, the clause node $c$ is connected to the three literal 
nodes corresponding to the literals that appear in clause $c$ in $\phi$.
Therefore, each literal node is connected to the two clauses in which it appears.
%
Graph $G(\phi)$ includes a {\em clique} of even size $N$, 
with $12C\leq N \leq \frac{95C}{16 \varepsilon}-\frac{123C}{4}$; 
the clique is disconnected from the rest of the graph.
%
Graph $G(\phi)$ includes $N+\frac{99C}{4}$ additional {\em isolated} nodes.
Overall, the total number of nodes in $G(\phi)$ is $n = 2N+\frac{99C}{4}+25V+18C= 2N+\frac{123C}{2}$.

A profile of initial preferences to the nodes of $G(\phi)$ is called \emph{proper} if:
for every variable $x$,
it assigns preference $1$ to node $w_0(x)$ and to exactly
one literal node of the gadget of $x$;
%
for every clause $c$,
it assigns preference $1$ to nodes $u_1(c)$ and $u_2(c)$;
%
it assigns preference $1$ to exactly $\frac{N}{2}$ nodes of
the clique;

it assigns preference $0$ to all the remaining nodes.
Hence, in a proper profile the number of nodes with preference $1$ is
$2V + 2C + \frac{N}{2} = \frac{7C}{2} + \frac{N}{2} \leq n(\frac{1}{4} - \varepsilon)$;
the inequality follows by the upper bound in the definition of $N$.

Theorem \ref{thm:reduction} will follow by the next two lemmas.
The first lemma proves that $G(\phi)$ has a proper profile of initial preferences 
that leads to a majority of nodes with preference $1$ if and only $\phi$ is satisfiable.
\begin{lemma}\label{lemma:proper}
The Boolean formula $\phi$ is satisfiable if and only 
if there exists a proper profile of initial preferences to the nodes 
of $G(\phi)$ so that a stable profile with at least $n/2$ nodes with preference $1$
can be reached by a sequence of updates.
\end{lemma}
\begin{proof}
First observe that every clique node switches her preference to $1$ 
(as the strict majority of its neighbors has initially preference $1$ and this 
number gradually increases until all clique nodes switch to $1$). 

We next prove that starting from a proper profile of initial preferences, 
there is a sequence of updates that leads to a stable profile in which
$17$ nodes of every variable gadget have preference $1$.
To see this, consider a proper profile that assigns preference $1$ to $x$
(and to $w_0(x)$) and the following sequence of updates:
node $v_1(x)$ switches from $0$ to $1$;
then, for $i=1, \ldots, 6$, node $v_{i+1}(x)$ switches to $1$ immediately after node $v_i(x)$;
node $v_0(x)$ switches to $1$ after node $v_7(x)$;
finally, $w_1(x), \ldots, w_{7}(x)$ can switch in any order. 
Observe that in this sequence any switching node is unhappy since has 
a strict majority of nodes with preference $1$ in its neighborhood. 
Also, the resulting profile where the $17$ nodes 
$w_0(x),w_1(x),\ldots,w_7(x)$, $v_0(x),v_1(x),\ldots,v_7(x)$, and $x$ have preference $1$ 
is stable, i.e., no node in the gadget is unhappy. 
Indeed, for each node with preference $1$, the strict majority of the preferences of its neighbors is $1$.
Hence, the node has no incentive to switch to preference $0$.
For each of the remaining nodes (with preference $0$), at least half of its neighbors is $0$.
Hence, this node has no incentive to switch to preference $1$ either.
A similar sequence can be constructed for a proper profile that assigns
preference $1$ to node $\overline{x}$ (and $w_0(x)$) of the gadget for variable $x$. 
Intuitively, the two proper profiles of initial preferences simulate the assignment of values {\sc True} and 
{\sc False} to variable $x$, respectively. 

In addition, it is easy to see that starting from a proper profile,
there is no sequence of updates that reaches a stable profile where 
more than $17$ nodes in a variable gadget have preference $1$ 
(the observation needed here is the same that guarantees that we reach a stable profile above).

Let us now consider the clause gadgets associated with clause $c$ of $\phi$.
We observe that, starting from a proper profile of initial preferences,
there exists a sequence of updates that leads to a stable profile
in which $17$ nodes of the clause gadget have preference $1$. 
Indeed, starting from the proper assignment of preference $1$ to nodes $u_1(c)$ and $u_2(c)$, 
nodes $\upsilon_1(c),\ldots,\upsilon_{15}(c)$ will switch from $0$ to $1$ in arbitrary order 
(for each of them both neighbors have preference $1$).
After these updates, at least $15$ out of $17$ neighbors of $u_1(c)$ and $u_2(c)$ have preference $1$ and both neighbors of the nodes $\upsilon_1(c), \ldots, \upsilon_{15}(c)$ have preference $0$. Hence, none among these nodes have any incentive to switch to preference $0$.


Let us now focus on the clause nodes and observe that node $c$ in the corresponding clause gadget will switch to $1$ if and only if at least one of the literal nodes corresponding
to literals that appear in $c$ have preference $1$ 
(since the degree of a clause node is five and nodes $u_1(c)$ and $u_2(c)$ have preference $1$). 
This switch cannot trigger any other switch in literal nodes or in nodes of clause gadgets since the preference of these nodes coincides with a strict majority of preferences in its neighborhood.
Hence, the fact that a clause node has preference $1$ (respectively, $0$) corresponds to the clause being satisfied (respectively, not satisfied) by the Boolean assignment induced by the proper profile of initial preferences. Eventually, the updates lead to an additional number of $C$ clause nodes adopting preference $1$ in the stable profile
if and only if $\phi$ is satisfiable.

In conclusion, we have that if $\phi$ is satisfiable there is a sequence of updates
converging to a profile with $17V + 17C + N + C = N+123C/4 = n/2$ nodes with preference $1$.
Otherwise, if $\phi$ is not satisfiable, any sequence of updates
converges to a stable profile with strictly less than $n/2$ nodes having preference $1$.
\end{proof}
We conclude the proof by showing that it is sufficient to restrict to proper assignments
as non-proper assignments will never lead to a stable profile with
a majority of nodes with preference $1$.
\begin{lemma}\label{lemma:non-proper}
For non-proper profiles that assign preference $1$ to at most $\frac{7C}{2} + \frac{N}{2}$ nodes,
there is no sequence of updates converging to a stable profile with at least $n/2$ nodes having preference $1$.
\end{lemma}
\begin{proof}
First observe that if the total number of clique and isolated nodes with preference $1$ is
strictly less than $\frac{N}{2}$, then no clique node with preference $0$ will adopt preference
$1$ (this holds trivially for isolated nodes too). Thus, in this case, even counting all nodes in variable and clause gadgets,
any sequence of updates will converge to a stable profile with at most $25V+18C+\frac{N}{2}-1 < \frac{n}{2}$ nodes with preference $1$ 
(the inequality follows since $N\geq 12C$).

Let us now focus on a profile of initial preferences that assigns preference $1$ to at most $7C/2=2C+2V$ nodes 
from variable and clause gadgets. Suppose that this profile of initial preferences is such that a sequence of 
updates leads to a stable profile with at least $n/2$ nodes with preference $1$.
We will show that this profile of initial preferences must be proper.

First, observe that if at most one node in a variable gadget is assigned preference $1$, 
then all nodes in the gadget will eventually adopt preference $0$ after a sequence of updates.
Indeed, a literal node will have at least six neighbors with preference $0$ 
and at most three with preference $1$, and any non-literal node will have at most 
one out of at least three of its neighbors with preference $0$. 

Consider now profiles of initial preferences that assign preference $1$ to two nodes 
of the variable gadget of $x$ in a non-proper way. 
We show that any sequence of updates leads to a profile in 
which all nodes of the gadget adopt preference $0$. 

Indeed, assume that $w_0(x)$ has preference $1$ and both $x$ and $\ox$ have preference $0$. 
Clearly, 
the nodes $w_1(x)$,\ldots, $w_7(x)$ can switch from $0$ to $1$ in any order. 
Among the non-literal nodes $v_i(x)$ and $v_i(\ox)$, only one among the degree-$3$ nodes $v_0(x)$, $v_1(x)$, and $v_1(\ox)$ can switch from $0$ to $1$; 
this can only happen if the second node with preference $1$ is in the neighborhood of one of these nodes (i.e., to some of the nodes $v_7(x)$, $v_7(\ox)$, $v_2(x)$, or $v_2(\ox)$). But then, the literal nodes will have at most four neighbors with preference $1$ and they cannot switch to $1$. So, no other node has any incentive to switch from $0$ to $1$. Then, $w_0(x)$ has at least $13$ among its $22$ neighbors with preference $0$ and will switch from $1$ to $0$, followed by the nodes $w_1(x)$, ..., $w_7(x)$ that will switch back to $0$ as well. Then, there are at most two nodes with preference $1$ among the nodes $v_i(x)$ and $v_i(\ox)$ that will eventually switch to $0$ as well (since they have degree at least three). 

Assume now that literal node $x$ has preference $1$ (the case for $\ox$ is symmetric) and 
that $w_0(x)$ has preference $0$. 
Then, only the degree-$3$ node $v_1(x)$ that is adjacent to $x$ can switch to $1$ provided that the second node with preference $1$ is node $v_2(x)$. Now notice that no other node can switch from $0$ to $1$. Even worse, the literal node $x$ has at least five (out of nine) neighbors of preference $0$ and will switch from $1$ to $0$. And then, we are left with at most two nodes with preference $1$ among the nodes $v_i(x)$ and $v_i(\ox)$ that will eventually switch to $0$ as well. 

Finally, we consider the case in which $w_0(x)$ and the two literal nodes have preference $0$. 
Now
the only node that can initially switch from $0$ to $1$ is $v_0(x)$ provided that the two nodes with preference $1$ are $v_7(x)$ and $v_7(\ox)$. But then, there is no other node that can switch from $0$ to $1$ and, eventually, nodes $v_7(x)$ and $v_7(\ox)$ will switch to $0$ and finally node $v_0(x)$ will switch back to $0$. 

We have covered all possible cases in which a variable gadget has a non-proper assignment of 
preference $1$ to two nodes and shown that in all of these cases, 
all nodes of the gadget will switch to preference $0$.
On the other hand, 
as discussed in the proof of Lemma~\ref{lemma:proper}, a proper profile of initial preferences can end up with preference $1$ in $17$ nodes of the variable gadget.

Now, observe that if at most one node in a clause gadget has preference $1$ 
(or two nodes are assigned preference $1$ in a non-proper way), 
then all the $17$ non-clause nodes in the gadget will end up with preference $0$. 
This is due to the fact that none among the nodes $\upsilon_1(c)$, ..., $\upsilon_{15}(c)$ can switch from $0$ to $1$ since at least one of their neighbors will have preference $0$. But this means that nodes $u_1(c)$ and $u_2(c)$ are adjacent to many (i.e., at least $13$) nodes with preference $0$; so, they will also switch to $0$. And then, if there is still some node $\upsilon_i(c)$ with preference $1$, it will switch to $0$ since both its neighbors have preference $1$. 

Now, by denoting with $V_0,V_1,V_3$  the number of variable gadgets that have $0,1$ or at least
$3$ nodes with preference $1$ and by $V_{2p}$ and $V_{2n}$ the number of variable gadgets
with proper and non-proper assignment of preference $1$ to exactly two nodes, we have
$V = V_0+V_1+V_{2n}+V_{2p}+V_3$
and, by denoting with $C_0$, $C_1$, and $C_3$ the number of clause gadgets with $0$, $1$, and at least $3$ nodes with preference $1$ in nodes other than the clause node and
by $C_{2p}$ and $C_{2n}$ the number of clause gadgets with two nodes with preference $1$
assigned in a proper and non-proper way, we have
$C = C_0+C_1+C_{2n}+C_{2p}+C_3$.
Since the total number of nodes with preference $1$ does not exceed $2V+2C$, we have
$V_1+2V_{2n}+2V_{2p}+3V_3+C_1+2C_{2n}+2C_{2p}+3C_3 \leq 2V+2C$
from which we get
$V_3+C_3 \leq 2C_0+C_1+2V_0+V_1$.
Now consider the difference between the number of nodes with preference $1$ in any stable profile reached after a sequence of updates and the quantity $17V+18C$. It is at most
$17V_{2p}+25V_3+C_0+C_1+C_{2n}+18C_{2p}+18C_3-17V-18C
= -17V_0-17V_1-17V_{2n}+8V_3-17C_0-17C_1-17C_{2n}
\leq -V_0-9V_1-17V_{2n}-C_0-9C_1-17C_{2n}-8C_3$.
Hence, if at least one of $V_0$, $V_1$, $V_{2n}$, $C_0$, $C_1$, $C_{2n}$, and $C_3$ is positive, the proof follows since the number of nodes with preference $1$ will be strictly less than $N+17V+18C=n/2$. Otherwise, i.e., if all these quantities are $0$, this implies that $C=C_{2p}$ and $V=V_{2p}+V_3$ which in turn implies that $V_3=0$ since the number of nodes with preference $1$ cannot exceed $2C+2V$. Hence, the only case where a sequence of updates may lead to a stable profile with at least $N+17V+18C$ nodes having preference $1$ is when the profile of initial preferences is proper. The claim follows.
\end{proof}
\noindent \emph{Checking whether minority can become majority.}
%
We next show that, given a graph $G$ and a profile of initial preferences $\blf$, 
it is possible to decide whether $\blf$ is \mbm\ for $G$ in polynomial time.
Moreover, if this is the case, then there is an efficient algorithm
that computes the subverting sequence of updates.
This algorithm was used in \cite{ckoEC13} for bounding the 
price of stability.
\begin{theorem}\label{thm:poly-mbm}
There is a polynomial time algorithm that, given a graph $G=(V,E)$ and a profile of initial preferences $\blf$, decides whether $\blf$ is \mbm\ for graph $G$ and, if it is, it outputs a subverting sequence of updates.
\end{theorem}
\begin{proof}
Consider the following algorithm that receives as input an $n$-node graph $G=(V,E)$
and a profile $\blf$ with less than $n/2$ nodes with initial preference $1$:\\
\textbf{Step 1:} While there exists an unhappy node $v$ whose current preference is $0$, update $v$'s preference to $1$.\\
\textbf{Step 2:} While there exists an unhappy node $v$ whose current preference is $1$, update $v$'s preference to $0$;
let $\blfp$ denote the profile at the end of this phase.\\
\textbf{Step 3:} If preference $1$ is a majority in $\blfp$, then return ``yes'' and output the sequence of updates of the first and second phase; otherwise, return ``no''.

Clearly, the running time of the algorithm is polynomial in the size of the input graph, since each node updates its preference at most twice. The fact that $\blfp$ is a stable profile is proved in \cite[Lemma~3.3]{ckoEC13}.
We next show that $\blfp$ is actually the stable profile that maximizes the number of nodes with preference $1$.
We refer to updates that change the preference from $x$ to $\ox$ as $x$-to-$\ox$ moves.

Consider a sequence $\sigma$ of updates leading to a stable profile $\st$ that maximizes the number of nodes with preference $1$.
We will show that there is another sequence of updates that has the form computed by the two-phase algorithm described above
and converges to a stable profile in which the agents with preference $1$ are at least as many as in $\st$.
We start by constructing a first sequence of moves $\sigma'$ from $\blf$ to $\st$ that contains at most two moves per agent and has the following properties:
$0$-to-$1$ moves are only executed by unhappy agents and precede the $1$-to-$0$ moves, which may or may not executed by unhappy agents.
To construct the sequence $\sigma'$, we repeatedly apply the following procedure to the sequence $\sigma$ until this is no longer possible:
pick an agent $i$ that performs an $1$-to-$0$ move just before the $0$-to-$1$ move of agent $i'$;
if $i=i'$ (this is possible when applying the process repeatedly), simply remove both moves from the sequence;
otherwise, swap the moves in the sequence.
The crucial observation is that $0$-to-$1$ moves are shifted to the beginning of the sequence in this way.
Note that such a move will be still executed by an unhappy node if it survives the application of the swap process
(because if the corresponding agent was unhappy after a $1$-to-$0$ move, it will also be unhappy also before it).
In contrast, the repeated application of the swap process shifts $1$-to-$0$ moves to the end of the sequence.

We now construct a sequence $\sigma''$ by keeping the $0$-to-$1$ moves as in $\sigma'$
and by completing the sequence with additional $1$-to-$0$ updates until reaching a stable profile.
We observe that the new sequence cannot contain any other $1$-to-$0$ move besides the ones in $\sigma'$
(otherwise, $\st$ would not be a stable profile).
Hence, if $\sigma$ has the maximum number of nodes with preference $1$, so does $\sigma''$.
The theorem follows by observing that $\sigma''$ has the form computed by the algorithm.
\end{proof}

\section{Conclusions and Open Problems}
\label{sec:conclusion}
In this work we showed that, for any social network topology
except very few and extreme cases,
social pressure can subvert a majority.
We proved this with respect to a very natural {\em majority} dynamics
in the case in which agents must express preferences.
We also showed that, for each of these graphs,
it is possible to compute in polynomial time
an initial majority and a sequence of updates that subverts it.
The initial majority constructed in this way
consists of only $\lceil (n+1)/2 \rceil$ agents.
On the other hand, our hardness result proves that
it may be hard to compute an initial majority of size at least $3n/4$
that can be subverted by the social pressure.
The main problem that this work left open is to close this gap.

Even if computational considerations rule out a simple characterization
of the graphs for which a large majority can be subverted,
it would be still interesting to gain knowledge on these graphs.
Specifically, can we prove that the set of graphs
for which large majority can be subverted
can be easily described by some simple (but hard to compute) graph-theoretic measure?
We believe that our ideas can be adapted (e.g., by considering unbalanced partitions in place of bisections),
for gaining useful hints in this direction.


\newpage
\bibliographystyle{plain}
\bibliography{opinion}

\newpage
\appendix

\end{document}